%% file: Paper.tex
\documentclass[a4paper,11pt]{article}

\usepackage{relsize}
\usepackage{epigraph}
\usepackage{nicefrac}
\usepackage[margin=1in]{geometry}
\usepackage{amsmath,amssymb,amsthm}
\usepackage{hyperref}
\hypersetup{colorlinks=true,linkcolor=black,citecolor=blue}

\usepackage{mathtools}
\usepackage{float}
\usepackage[titletoc,title]{appendix}
\usepackage{graphicx, color}
\usepackage[classfont=sanserif,langfont=roman,funcfont=italic]{complexity}
\usepackage{caption}
\usepackage{subcaption}
\usepackage{setspace}
\floatstyle{ruled}
\newfloat{algorithm}{ht}{loa}
\providecommand{\algorithmname}{Algorithm}
\floatname{algorithm}{\protect\algorithmname}

\newcommand{\Supp}{\textsl{Supp}}

\newtheorem{theorem}{Theorem}[section]

\newtheorem{lemma}[theorem]{Lemma}

\newtheorem{claim}[theorem]{Claim}
\newtheorem{proposition}[theorem]{Proposition}
\newtheorem{definition}[theorem]{Definition}
\newtheorem{remark}[theorem]{Remark}

\def\bits{ \{0,1\} }

\def\Ex{{\mathbb {E}}}
\def\N{{\mathbb {N}}}

\def\R{{\mathbb {R}}}
\def\N{{\mathbb {N}}}

\def\poly{{\rm {poly}}}
\def\polylog{{\rm {polylog}}}

\def\eps{\varepsilon}

\title{Communication Complexity of Correlated Equilibrium in Two-Player Games}
\author{Anat Ganor%
\thanks{Tel-Aviv University, Israel.
	 The research leading to these results has received funding from the Israel Science Foundation (grant number 552/16)
     and the I-CORE Program of the planning and budgeting committee and The Israel Science Foundation (grant number 4/11).
     Email:anat.ganor@gmail.com}
\and Karthik C.\ S.\thanks{Weizmann Institute of Science, Israel.
   This work was partially supported by ISF-UGC 1399/14 grant.
    Email:karthik.srikanta@weizmann.ac.il}}

\date{}

\begin{document}
\maketitle

\begin{abstract}
\input{abstract}
\end{abstract}
\clearpage

\tableofcontents
\clearpage

\section{Introduction}
\input{Introduction}

\section{Preliminaries}

\subsection{General Notation}

\paragraph{Strings.} For two bit strings $x,y\in\bits^*$, let $xy$ be the concatenation of $x$ and $y$.
For a bit string $x\in\bits^n$ and an index $i\in[n]$, $x_i$ is the $i^{th}$ bit in $x$ 
and $\bar{x}$ is the negated bit string, that is $\bar{x}_i$ is the negation of $x_i$.

\paragraph{Probabilities of sets.} For a function $\mu:\Omega\rightarrow[0,1]$, where $\Omega$ is some finite set, 
and a subset $S\subseteq\Omega$, let $\mu(S) = \sum_{z\in S}\mu(z)$.
Define $\mu(\emptyset) = 0$ and $\max_{z\in \emptyset}\mu(z) = 0$.
For a function $\mu:\mathcal{U}\times\mathcal{V}\rightarrow[0,1]$, 
where $\mathcal{U},\mathcal{V}$ are some finite sets, 
and a subset $S\subseteq\mathcal{U}\times\mathcal{V}$,
let $\mu(S) = \sum_{(u,v)\in S}\mu(u,v)$.
For a subset $S\subseteq\mathcal{U}$ and $v\in\mathcal{V}$,
let $\mu(S,v) = \sum_{u\in S}\mu(u,v)$.
Similarly, for a subset $S\subseteq\mathcal{V}$ and $u\in \mathcal{U}$
let $\mu(u,S) = \sum_{v\in S}\mu(u,v)$.

\paragraph{Conditional distributions.} For a distribution $\mu$ over $\mathcal{U}\times\mathcal{V}$,
where $\mathcal{U},\mathcal{V}$ are some finite sets, 
and $u\in\mathcal{U}$ let $\mu|u$ be the distribution over $\mathcal{V}$ defined as
$$ \mu|u(v) = \Pr_{(u',v')\sim\mu}[v=v' ~|~ u=u'] ~~~~\forall~ v\in\mathcal{V} .$$
Similarly, for $v\in\mathcal{V}$ let $\mu|v$ be the distribution over $\mathcal{U}$ defined as
$$ \mu|v(u) = \Pr_{(u',v')\sim\mu}[u=u' ~|~ v=v'] ~~~~\forall~ u\in\mathcal{U} .$$

\subsection{Win-Lose and Bayesian Games}

A win-lose, finite game for two players $A$ and $B$ 
is given by two utility functions
$u_A:\mathcal{U}\times\mathcal{V}\rightarrow\bits$ 
and $u_B:\mathcal{U}\times\mathcal{V}\rightarrow\bits$,
where $\mathcal{U}$ and $\mathcal{V}$ are finite sets of actions.
We say that the game is an $N\times N$ game, where $N=\max\{|\mathcal{U}|,|\mathcal{V}|\}$. 
A \emph{mixed strategy} for player $A$ is a distribution over $\mathcal{U}$
and a mixed strategy for player $B$ is a distribution over $\mathcal{V}$.
A mixed strategy is called \emph{pure} if it has only one action in its support. 
A \emph{correlated mixed strategy} is a distribution over $\mathcal{U}\times\mathcal{V}$.
A \emph{switching rule} for player $A$ is a mapping from $\mathcal{U}$ to $\mathcal{U}$
and a switching rule for player $B$ is a mapping from $\mathcal{V}$ to $\mathcal{V}$.

A Bayesian, finite game for two players $A$ and $B$ 
is given by a distribution $\phi$ over $\Theta_A\times\Theta_B$ 
and two utility functions
$u_A:\Theta_A\times\Sigma_A\times\Sigma_B\rightarrow[0,1]$ 
and $u_B:\Theta_B\times\Sigma_A\times\Sigma_B\rightarrow[0,1]$,
where $\Sigma_A$, $\Sigma_B$ are finite sets of actions,
and $\Theta_A$, $\Theta_B$ are finite sets of types.
We say that the game is on $N$ actions and $T$ types 
where $N = \max\{|\Sigma_A|,|\Sigma_B|\}$ and $T = \max\{|\Theta_A|, |\Theta_B|\}$. 
A \emph{mixed strategy} for player $A$ is a distribution over $\Sigma_A$
and a mixed strategy for player $B$ is a distribution over $\Sigma_B$.
A mixed strategy is called \emph{pure} if it has only one action in its support.

\subsection{Approximate Correlated Equilibrium}

\begin{definition}\label{def:ACE}
Let $\eps\in[0,1)$.
An $\eps$-approximate correlated equilibrium of a two-player game
is a correlated mixed strategy $\mu$ 
such that the following two conditions hold:
\begin{enumerate}
\item For every actions $u,u'\in\mathcal{U}$, 
$$ \sum_{v\in\mathcal{V}} \mu(u,v) \cdot \left( u_A(u',v) - u_A(u,v) \right) \leq \eps .$$
\item For every actions $v,v'\in\mathcal{V}$, 
$$ \sum_{u\in\mathcal{U}} \mu(u,v) \cdot \left( u_B(u,v') - u_B(u,v) \right) \leq \eps .$$
\end{enumerate}
\end{definition}

\begin{definition}
Let $\eps\in[0,1)$.
An $\eps$-approximate rule correlated equilibrium of a two-player game
is a correlated mixed strategy $\mu$ 
such that the following two conditions hold:
\begin{enumerate}
\item For every switching rule $f$ for player $A$, 
$$ \Ex_{(u,v)\sim\mu} \left[ u_A(f(u),v) - u_A(u,v) \right] \leq \eps .$$
\item For every switching rule $f$ for player $B$, 
$$ \Ex_{(u,v)\sim\mu} \left[ u_B(u,f(v)) - u_B(u,v) \right] \leq \eps .$$
\end{enumerate}
\end{definition}

When the approximation value is zero, the two notions above coincide.
The following proposition states that every approximate rule correlated equilibrium 
is an approximate correlated equilibrium.

\begin{proposition}\label{ARCEtoACE}
Let $\eps\in[0,1)$ and let $\mu$ be an $\eps$-approximate rule correlated equilibrium of a two-player game.
Then, $\mu$ is an $\eps$-approximate correlated equilibrium of the game.
\end{proposition}
\begin{proof}
Assume towards a contradiction that $\mu$ is not an $\eps$-approximate correlated equilibrium.
Then, without loss of generality, there exist actions $u',u''\in\mathcal{U}$ such that
$$ \sum_{v\in\mathcal{V}} \mu(u',v) \cdot \left( u_A(u'',v) - u_A(u',v) \right) > \eps .$$
Define a switching rule for player $A$ as follows:
$$ f(u)= \begin{cases*}u'' &if $u=u'$ \\u &otherwise\end{cases*} .$$
Then, for every $v\in\mathcal{V}$ and $u\in\mathcal{U}\setminus\{u'\}$, 
it holds that $u_A(f(u),v) - u_A(u,v) = 0$.
Therefore,
\begin{align*}
\Ex_{(u,v)\sim\mu} \left[ u_A(f(u),v) - u_A(u,v) \right] 
&= \sum_{v\in\mathcal{V}} \mu(u',v)\cdot \left( u_A(f(u'),v) - u_A(u,v) \right) 
> \eps ,
\end{align*}
which is a contradiction.
\end{proof}

In the other direction, the following holds.

\begin{proposition}\label{ACEtoARCE}
Let $\eps\in[0,1)$ and let $\mu$ be an $\eps$-approximate correlated equilibrium of an $N$ action two-player game.
Then, $\mu$ is an $(\eps\cdot N)$-approximate rule correlated equilibrium of the game.
\end{proposition}
\begin{proof}
Assume towards a contradiction that $\mu$ is not an $(\eps\cdot N)$-approximate rule correlated equilibrium. 
Then, without loss of generality, there exists a function $f:\mathcal{U}\to \mathcal{U}$ such that 
$$ \Ex_{(u,v)\sim\mu} \left[ u_A(f(u),v) - u_A(u,v) \right] > \eps\cdot N.$$
From an averaging argument, there exists $u'\in\mathcal{U}$ such that 
$$ \sum_{v\in\mathcal{V}} \mu(u',v) \cdot \left( u_A(f(u'),v) - u_A(u',v) \right) > \eps $$
which is a contradiction.
\end{proof}

\subsubsection*{The Communication Task}

The communication task of finding an $\eps$-approximate (rule) correlated equilibrium
is as follows.
Consider a win-lose, finite game for two players $A$ and $B$, 
given by two utility functions
$u_A:\mathcal{U}\times\mathcal{V}\rightarrow\bits$ 
and $u_B:\mathcal{U}\times\mathcal{V}\rightarrow\bits$.

\paragraph{The inputs:}
The actions sets $\mathcal{U},\mathcal{V}$ and the approximation value $\eps$ are known to both players.
Player $A$ gets the utility function $u_A$ 
and player $B$ gets the utility function $u_B$.
The utility functions are given as truth tables of size $|\mathcal{U}|\times|\mathcal{V}|$ each.

\paragraph{At the end of the communication:}
Both players know the same correlated mixed strategy $\mu$
over $\mathcal{U}\times\mathcal{V}$,
such that $\mu$ is an $\eps$-approximate (rule) correlated equilibrium.

\subsection{Approximate Nash Equilibrium}

\begin{definition}\label{def:ANE}
Let $\eps\in[0,1)$.
An $\eps$-approximate Nash equilibrium of a two-player game 
is a pair of mixed strategies $(a^*,b^*)$ for the players $A,B$ respectively, 
such that the following two conditions hold:
\begin{enumerate}
\item For every mixed strategy $a$ for player $A$, 
$$ \Ex_{u\sim a,v\sim b^*} [ u_A(u,v) ] - \Ex_{u\sim a^*,v\sim b^*} [ u_A(u,v) ] \leq \eps .$$
\item For every mixed strategy $b$ for player $B$, 
$$ \Ex_{u\sim a^*,v\sim b} [ u_B(u,v) ] - \Ex_{u\sim a^*,v\sim b^*} [ u_B(u,v) ] \leq \eps .$$
\end{enumerate}
\end{definition}

An approximate Nash equilibrium with $\eps=0$ is called an exact Nash equilibrium.
A two-player game has a pure Nash equilibrium $(u,v)$, 
where $u\in\mathcal{U}$ and $v\in\mathcal{V}$, 
if there exists an exact Nash equilibrium $(a^*,b^*)$,
where the support of $a^*$ is $\{u\}$ and the support of $b^*$ is $\{v\}$.

\begin{definition}\label{def:AWSNE}
Let $\eps\in[0,1)$.
An $\eps$-approximate well supported Nash equilibrium of a two-player game
is a pair of mixed strategies $(a^*,b^*)$ for the players $A,B$ respectively, 
such that the following two conditions hold:
\begin{enumerate}
\item For every action $u\in \Supp(a^*)$ and every action $u'\in\mathcal{U}$, 
$$ \Ex_{v\sim b^*} [ u_A(u',v) - u_A(u,v) ] \leq \eps .$$
\item For every action $v\in \Supp(b^*)$ and every action $v'\in\mathcal{V}$, 
$$ \Ex_{u\sim a^*} [ u_B(u,v') - u_B(u,v) ] \leq \eps .$$
\end{enumerate}
\end{definition}

When the approximation value is zero, the two notions above coincide.
The following proposition states that every approximate well supported Nash equilibrium 
is an approximate Nash equilibrium. 

\begin{proposition}\label{WSNEtoNE}
Let $\eps\in[0,1)$ and let $(a^*,b^*)$ be an $\eps$-approximate well supported Nash equilibrium of a two-player game.
Then, $(a^*,b^*)$ is an $\eps$-approximate Nash equilibrium of the game.
\end{proposition}
\begin{proof}
Let $a$ be a mixed strategy for player $A$.
For every action $u\in \Supp(a^*)$ and every action $u'\in \Supp(a)$, 
$$ \Ex_{v\sim b^*} [ u_A(u',v) - u_A(u,v) ] \leq \eps .$$
Therefore,
\begin{align*}
\eps &\geq \Ex_{u'\sim a, u\sim a^*}\Ex_{v\sim b^*} [ u_A(u',v) - u_A(u,v) ] \\
&= \Ex_{u'\sim a,v\sim b^*} [ u_A(u',v) ] - \Ex_{u\sim a^*,v\sim b^*} [ u_A(u,v) ].
\end{align*}
Similarly, for every mixed strategy $b$ for player $B$, 
every action $v\in \Supp(b^*)$ and every action $v'\in \Supp(b)$, 
$$ \Ex_{u\sim a^*} [ u_B(u,v') - u_B(u,v) ] \leq \eps .$$
Therefore,
\begin{align*}
\eps &\geq \Ex_{v'\sim b, v\sim b^*}\Ex_{u\sim a^*} [ u_B(u,v') - u_B(u,v) ] \\
&= \Ex_{u\sim a^*,v'\sim b} [ u_B(u,v) ] - \Ex_{u\sim a^*,v\sim b^*} [ u_B(u,v) ].
\end{align*}
\end{proof}

\subsubsection*{The Communication Task}

The communication task of finding an $\eps$-approximate (well supported) Nash equilibrium
Consider a win-lose, finite game for two players $A$ and $B$, 
given by two utility functions
$u_A:\mathcal{U}\times\mathcal{V}\rightarrow\bits$ 
and $u_B:\mathcal{U}\times\mathcal{V}\rightarrow\bits$.

\paragraph{The inputs:}
The actions sets $\mathcal{U},\mathcal{V}$ and the approximation value $\eps$ are known to both players.
Player $A$ gets the utility function $u_A$ 
and player $B$ gets the utility function $u_B$.
The utility functions are given as truth tables of size $|\mathcal{U}|\times|\mathcal{V}|$ each.

\paragraph{At the end of the communication:}
Player $A$ knows a mixed strategy $a^*$ over $\mathcal{U}$
and player $B$ knows a mixed strategy $b^*$ over $\mathcal{V}$,
such that $(a^*,b^*)$ is an $\eps$-approximate (well supported) Nash equilibrium.

\subsection{Approximate Bayesian Nash Equilibrium}

\begin{definition}
An $\eps$-approximate Bayesian Nash equilibrium of a two-player game
is a set of mixed strategies $\{a^*_{t_A}\}_{t_A\in\Theta_A}$ for player $A$
and a set of mixed strategies $\{b^*_{t_B}\}_{t_B\in\Theta_B}$ for player $B$
such that the following two conditions hold:
\begin{enumerate}
\item For every type $t_A\in\Theta_A$ and every mixed strategy $a$ for player $A$, 
$$ \Ex_{t_B \sim \phi|t_A} \Ex_{u\sim a, v\sim b^*_{t_B}} [ u_A(t_A,u,v) ] 
 - \Ex_{t_B \sim \phi|t_A} \Ex_{u\sim a^*_{t_A}, v\sim b^*_{t_B}} [ u_A(t_A,u,v) ] \leq \eps .$$
\item For every type $t_B\in\Theta_B$ and every mixed strategy $b$ for player $B$, 
$$ \Ex_{t_A \sim \phi|t_B} \Ex_{u\sim a^*_{t_A}, v\sim b} [ u_B(t_B,u,v) ] 
 - \Ex_{t_A \sim \phi|t_B} \Ex_{u\sim a^*_{t_A}, v\sim b^*_{t_B}} [ u_B(t_B,u,v) ] \leq \eps .$$
\end{enumerate}
\end{definition}

\subsubsection*{The Communication Task}

The communication task of finding an $\eps$-approximate Bayesian Nash equilibrium
Consider a Bayesian, finite game for two players $A$ and $B$, on $N$ actions and $T$ types, 
that is given by the distribution $\phi$ over $\Theta_A\times\Theta_B$ 
and the two utility functions
$u_A:\Theta_A\times\Sigma_A\times\Sigma_B\rightarrow[0,1]$ 
and $u_B:\Theta_B\times\Sigma_A\times\Sigma_B\rightarrow[0,1]$.

\paragraph{The inputs:}
All the sets $\Theta_A, \Theta_B, \Sigma_A, \Sigma_B$, the distribution $\phi$ and the approximation value $\eps$ 
are known to both players.
Player $A$ gets the utility function $u_A$ 
and player $B$ gets the utility function $u_B$.
The utility functions are given as truth tables of size at most $T\cdot N^2$ each.

\paragraph{At the end of the communication:}
Player $A$ knows a set of mixed strategies $\{a^*_{t_A}\}_{t_A\in\Theta_A}$
and player $B$ knows a set of mixed strategies $\{b^*_{t_B}\}_{t_B\in\Theta_B}$,
such that $\left(\{a^*_{t_A}\}_{t_A\in\Theta_A},\{b^*_{t_B}\}_{t_B\in\Theta_B}\right)$ 
is an $\eps$-approximate Bayesian Nash equilibrium.

%
%

\section{The 2-Cycle Game}\label{sec:the-game}

Let $n\in\N$, $n\geq 3$.
The 2-cycle game is constructed from two $n$-bit strings $x,y\in\bits^n$
for which there exists exactly one index $i\in[n]$, such that $x_{i} > y_{i}$.
Throughout the paper, all operations (adding and subtracting) are done modulo $n$.

\paragraph{The graphs.}
Given a string $x\in\{0,1\}^n$, player $A$ computes the graph $G_A$ 
on the set of vertices $V=[n]\times\{0,1,01,11\}$ with the following set of directed edges
(an edge $(u,v)$ is directed from $u$ into $v$):
\begin{align*}
E_A &= \mathlarger{\mathlarger{\Bigg\lbrace}}((i,1),(i+1, z)): i\in[n], 
		z=\begin{cases*}0 &if $x_{i+1}=0$ \\11 &otherwise\end{cases*}\mathlarger{\mathlarger{\Bigg\rbrace}}\\
  &\mathlarger{\mathlarger{\mathlarger\cup}}\mathlarger{\mathlarger{\Bigg\lbrace}}((i,0),(i+1, z)): i\in[n], x_i=0, 
  		z=\begin{cases*}0 &if $x_{i+1}=0$ \\01 &otherwise\end{cases*}\mathlarger{\mathlarger{\Bigg\rbrace}}\\
  &\mathlarger{\mathlarger{\mathlarger\cup}}\Big\lbrace ((i,0),(i-1, x_{i-1})): x_i=1, i\in[n] \Big\rbrace\\
  &\mathlarger{\mathlarger{\mathlarger\cup}}\Big\lbrace ((i,z1),(i, 1)): z\in\bits, i\in[n] \Big\rbrace.
\end{align*}

See an example of such a graph in Figure~\ref{fig1a}.
\begin{figure}[!ht]
	\begin{center}
		\resizebox{\linewidth}{!}{
		\includegraphics{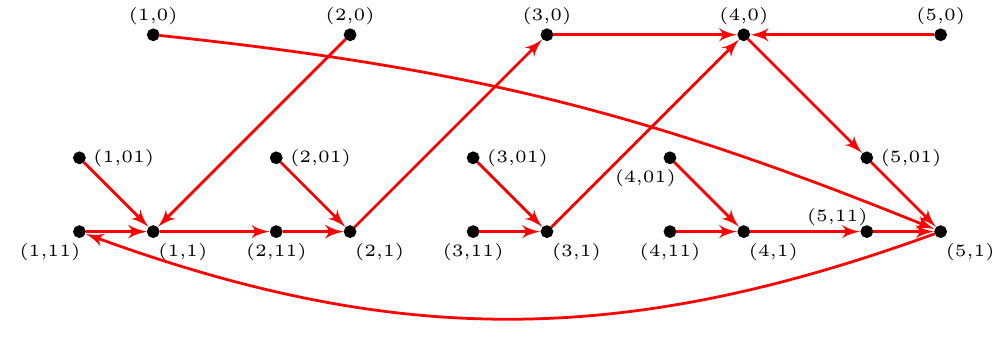}}
	\end{center}
		\caption{ The graph $G_A$ built from the 5 bit string 11001.}
		\label{fig1a}
\end{figure}

Given a string $y\in\{0,1\}^n$, player $B$ computes the graph $G_B$ 
on the same set of vertices $V$ with the following set of directed edges:
\begin{align*}
E_B &= \mathlarger{\mathlarger{\Bigg\lbrace}}((i,1),(i+1, z)): i\in[n], 
		z=\begin{cases*}0 &if $y_{i+1}=0$ \\11 &otherwise\end{cases*}\mathlarger{\mathlarger{\Bigg\rbrace}}\\
  &\mathlarger{\mathlarger{\mathlarger\cup}}\mathlarger{\mathlarger{\Bigg\lbrace}}((i,0),(i+1, z)): i\in[n],
  		z=\begin{cases*}0 &if $y_{i+1}=0$ \\01 &otherwise\end{cases*}\mathlarger{\mathlarger{\Bigg\rbrace}}\\
  &\mathlarger{\mathlarger{\mathlarger\cup}}\Big\lbrace ((i,z1),(i, 1)): z\in\bits, i\in[n] \Big\rbrace.
\end{align*}

See an example of such a graph in Figure~\ref{fig1b}.
\begin{figure}[!ht]
	\begin{center}
		\resizebox{\linewidth}{!}{
		\includegraphics{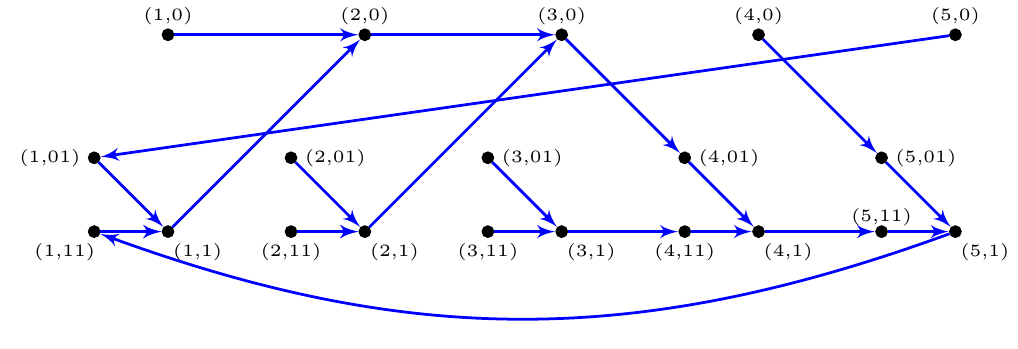}}
	\end{center}
   		\caption{ The graph $G_B$ built from the 5 bit string 10011.}
   		\label{fig1b}
\end{figure}

\paragraph{The actions and utility functions.}
The sets of actions are $\mathcal{U}=\mathcal{V}=V$. 
The utility function $u_A:V^2\rightarrow\bits$ of player $A$ is defined for every pair of actions $(u,v)\in V^2$ as 
$$ u_A(u,v)= \begin{cases*}1 &if $(v,u)\in E_A$ \\0 &otherwise\end{cases*} .$$
The utility function $u_B:V^2\rightarrow\bits$ of player $B$ is defined for every pair of actions $(u,v)\in V^2$ as 
$$ u_B(u,v)= \begin{cases*}1 &if $(u,v)\in E_B$ \\0 &otherwise\end{cases*} .$$
This is a win-lose, $N\times N$ game, where $N=4n$. 
We call it the \emph{2-cycle game} or more specifically, the \emph{2-cycle $N\times N$ game}.

\subsection{Notations}\label{sec:the-game-notations}

For two vertices $u,v \in V$, $(u,v)$ is a \emph{2-cycle}
if $(v,u)\in E_A$ and $(u,v)\in E_B$.
For a vertex $u\in V$, define 
\begin{align*}
& N_A(u)=\{v\in V~:~ (v,u)\in E_A \} \\
& N_B(u)=\{v\in V~:~ (v,u)\in E_B \}.
\end{align*}
That is, $N_A(u)$ is the set of incoming neighbors to $u$ in $E_A$, 
and $N_B(u)$ is the set of incoming neighbors to $u$ in $E_B$.
Let $d_A(u) = \left|N_A(u)\right|$ and $d_B(u) = \left|N_B(u)\right|$.
For a subset $S\subseteq V$, define
\begin{align*}
& N_A(S) = \cup_{v\in S}N_A(v) \\
& N_B(S) = \cup_{v\in S}N_B(v) .
\end{align*}
For every $i\in[n]$, \emph{layer $i$} is the set of vertices defined as
$$ L_i = \{(i,z) ~:~ z\in\{0,1\}\} .$$ 
Another useful set of vertices is a \emph{midway layer} defined as
$$ L^{m}_i = \{(i,z) ~:~ z\in\{01,11,0\}\} .$$
Edges in $E_A$ of the form $((i,0),(i-1, x_{i-1}))$ for $i\in[n]$ are called \emph{back-edges}.
For a vertex $u\in V$, define 
\begin{align*}
& N_A^f(u)=\{v\in V~:~ (v,u)\in E_A \text{  and $(v,u)$ is not a back-edge } \}.
\end{align*}
For a subset $S\subseteq V$, define
\begin{align*}
& N_A^f(S) = \cup_{v\in S}N_A^f(v) .
\end{align*}
Let $x,y$ be the strings from which the game was constructed.
Note that $u_A$ determines $x$, and $u_B$ determines $y$.
For an index $i\in[n]$ we say that $i$ is \emph{disputed} if $x_{i} > y_{i}$.
Otherwise, we say that $i$ is \emph{undisputed}.
Define $i^*$ to be the unique disputed index.
We denote the following key vertices:
\begin{align*}
& u^*      = (i^*-1,x_{i^*-1}) \\
& v^*_0    = (i^*,0) \\
& v^*_1    = (i^*,1) \\
& v^*_{01} = (i^*,01) \\
& v^*_{11} = (i^*,11) .
\end{align*}
For $p\in[0,1]$, $u\in V$ and a function $f:V\rightarrow[0,1]$, 
we say that $f$ is \emph{$p$-concentrated on $u$} if 
\begin{align*}
& f(v) \leq p ~~~~\forall~ v\in V\setminus\{u\} .
\end{align*}
To simplify notations, for a function $f$ taking inputs from the set $V$ 
and a vertex $v=(i,z)\in V$, we write $f(i,z)$ instead of $f((i,z))$.

\subsection{Basic Properties}\label{sec:the-game-properties}

The following are some useful, basic properties of the 2-cycle game.

\begin{proposition}[Out-degree]\label{prop:nonzero-out-degree}
For every $v\in V$, there exists exactly one $u\in V$ 
such that $u_A(u,v)=1$.
Similarly, for every $u\in V$, there exists exactly one $v\in V$
such that $u_B(u,v)=1$.
\end{proposition}
\begin{proof}
Follows immediately from the definitions of $E_A$ and $E_B$.
\end{proof}

\begin{proposition}[Max in-degree]\label{prop:max-in-degree}
For every $v\in V$, it holds that $d_A(v)\leq 3$ and $d_B(v)\leq 2$.
\end{proposition}
\begin{proof}
Follows immediately from the definitions of $E_A$ and $E_B$.
\end{proof}

By the following claim, the 2-cycle game has exactly one 2-cycle.

\begin{proposition}[A 2-cycle]\label{prop:2-cycle}
Let $(v,u)\in E_A$ be a back-edge.
If $v\neq v^*_0$, then $d_B(v)=0$.
Otherwise, $u=u^*$ and $(u^*,v^*_0)$ is a 2-cycle.
\end{proposition}
\begin{proof}
Let $u=(i,z_A)\in V$, for some $z_A\in\{0,1,01,11\}$
and assume there exits $v=(i+1,z_B)\in N_A(u)$, for some $z_B\in\{0,1,01,11\}$.
By the definition of $E_A$, 
$$ z_A=x_i, ~~~x_{i+1}=1 ~~~\text{and}~~~ z_B=0 .$$
If $v\neq v^*_0$, 
then $y_{i+1}=1$ and by the definition of $E_B$, $d_B(v)=0$.
Otherwise $v = v^*_0$ and $x_{i^*} > y_{i^*}$.
Since $v = v^*_0$ it holds that $u=u^*$.
Since $x_{i^*} > y_{i^*}$ it holds that $y_{i+1}=0$ 
and by the definition of $E_B$, $(u,v)\in E_B$.
\end{proof}

\subsection{Pure Nash Equilibrium}\label{sec:pure-nash}

By Claim~\ref{clm:pure-nash} below, 
the 2-cycle game has a unique pure Nash equilibrium.
Together with Proposition~\ref{prop:2-cycle}, the pure Nash equilibrium of the game is its 2-cycle.

\begin{claim}\label{clm:pure-nash}
The 2-cycle game has exactly one pure Nash equilibrium $(u^*,v^*_0)$.
\end{claim}
\begin{proof}
By Proposition~\ref{prop:2-cycle}, $(u^*,v^*_0)$ is a 2-cycle.
That is, $u_A(u^*,v^*_0)=1$ and $u_B(u^*,v^*_0)=1$.
Let $a^*$ be the mixed strategy for player $A$ with support $\{u^*\}$ and
$b^*$ be the mixed strategy for player $B$ with support $\{v^*_0\}$.
Then, 
$$ \Ex_{u\sim a^*,v\sim b^*} [ u_A(u,v) ] = \Ex_{u\sim a^*,v\sim b^*} [ u_B(u,v) ] = 1 .$$
For every mixed strategy $a$ for player $A$ it holds that
$$ \Ex_{u\sim a,v\sim b^*} [ u_A(u,v) ] = \Ex_{u\sim a} [ u_A(u,v^*_0) ] \leq 1. $$
Similarly, for every mixed strategy $b$ for player $B$,
$$ \Ex_{u\sim a^*,v\sim b} [ u_B(u,v) ] = \Ex_{v\sim b} [ u_B(u^*,v) ] \leq 1. $$
Therefore, $(u^*,v^*_0)$ is a pure Nash equilibrium.
\\
Let $u,v\in V$ such that $u\neq u^*$ or $v\neq v^*_0$. 
Let $a'$ be the mixed strategy for player $A$ with support $\{u\}$ and
$b'$ be the mixed strategy for player $B$ with support $\{v\}$.
By Proposition~\ref{prop:2-cycle}, either $(v,u)\notin E_A$ or $(u,v)\notin E_B$.
By proposition~\ref{prop:nonzero-out-degree}, 
there exist $u',v'\in V$ such that $(v,u')\in E_A$ and $(u,v')\in E_B$.
If $(v,u)\notin E_A$ then let $a$ be the mixed strategy for player $A$ with support $\{u'\}$.
We get that 
$$ \Ex_{u''\sim a,v''\sim b'} [ u_A(u'',v'') ] = u_A(u',v) = 1 $$
and
$$ \Ex_{u''\sim a',v''\sim b'} [ u_A(u'',v'') ] = u_A(u,v) = 0. $$
Otherwise $(u,v)\notin E_B$, then let $b$ be the mixed strategy for player $B$ with support $\{v'\}$.
We get that 
$$ \Ex_{u''\sim a',v''\sim b} [ u_B(u'',v'') ] = u_B(u,v') = 1 $$
and
$$ \Ex_{u''\sim a',v''\sim b'} [ u_B(u'',v'') ] = u_B(u,v) = 0. $$
Therefore, $(u,v)$ is not a pure Nash equilibrium.
\end{proof}


The following theorem states that 
finding the pure Nash equilibrium (equivalently, the 2-cycle) of the 2-cycle game is hard.
The proof is by a reduction from the following search variant of unique set disjointness:
Player $A$ gets a bit string $x\in\bits^n$ and player $B$ gets a bit string $y\in\bits^n$.
They are promised that there exists exactly one index $i^*\in[n]$ such that $x_{i^*} > y_{i^*}$.
Their goal is to find the index $i^*$.
It is well known that the randomized communication complexity of solving this problem 
with constant error probability is $\Omega(n)$ \cite{BFS86, KS92, Raz92}.
This problem is called the \emph{universal monotone relation}. 
For more details on the universal monotone relation and its connection to unique set disjointness see \cite{KN97}.

\begin{theorem}\label{thm:pure}
Every randomized communication protocol 
for finding the pure Nash equilibrium 
of the 2-cycle $N\times N$ game, 
with error probability at most $\frac{1}{3}$, 
has communication complexity at least $\Omega(N)$.
\end{theorem}
\begin{proof}
Let $x,y\in \bits^n$ be the inputs to the search variant of unique set disjointness described above.
Consider the 2-cycle $N\times N$ game which is constructed from these inputs, 
given by the utility functions $u_A,u_B$.
Assume towards a contradiction that there exists a communication protocol $\pi$ for finding 
the pure Nash equilibrium of the 2-cycle game with error probability at most $\nicefrac{1}{3}$
and communication complexity $o(N)$.
The players run $\pi$ on $u_A,u_B$ and with probability at least $\nicefrac{2}{3}$,
at the end of the communication, player $A$ knows $u$ and player $B$ knows $v$, 
such that $(u,v)$ is the pure Nash equilibrium of the game.
By Claim~\ref{clm:pure-nash}, $u=u^*$ and $v=v^*_0$.
Given $u^*,v^*_0$ to the players $A$ and $B$ respectively, 
both players know the index $i^*$, which is a contradiction.
\end{proof}

\section{Approximate Correlated Equilibrium of The 2-Cycle Game}\label{sec:ACE}

The following theorem states that 
given an approximate correlated equilibrium of the 2-cycle game, 
the players can recover the pure Nash equilibrium. 

\begin{theorem}\label{thm:ACEtoPure}
Consider a 2-cycle $N\times N$ game, given by the utility functions $u_A,u_B$.
Let $\eps\leq\frac{1}{24N^3}$ and 
let $\mu$ be an $\eps$-approximate correlated equilibrium of the game.
Then, there exists a deterministic communication protocol, 
that given $u_A$, $\mu$ to player $A$, and $u_B$, $\mu$ to player $B$,
uses $O(\log N)$ bits of communication, and at the end of the communication
player $A$ outputs an action $u\in V$ and player $B$ outputs an action $v\in V$, 
such that $(u,v)$ is the pure Nash equilibrium of the game.
\end{theorem}

Theorem~\ref{thm:ACE} follows from Theorem~\ref{thm:pure} and Theorem~\ref{thm:ACEtoPure}.
For the rest of this section we prove Theorem~\ref{thm:ACEtoPure}.
The proof uses the notion of slowly increasing probabilities. 
For more details on slowly increasing probabilities see Section~\ref{sec:on-slowly-increasing-probabilities}.

\begin{definition}[Slowly increasing probabilities]\label{def:slowly-increasing-probabilities}
Let $\delta\in[0,1]$.
A pair of functions $a:V\rightarrow[0,1]$ and $b:V\rightarrow[0,1]$ 
is \emph{$\delta$-slowly increasing} for the 2-cycle game 
if for every $u\in V$ the following conditions hold:
\begin{enumerate}
\item $a(u)\leq b(N_A(u))+\delta$.
\item $b(u)\leq a(N_B(u))+\delta$.
\end{enumerate}
\end{definition}
\noindent In particular, if $d_A(u)=0$ then $a(u)\leq b(\emptyset) + \delta = \delta$.
Similarly, if $d_B(u)=0$ then $b(u)\leq \delta$.

The next lemma states that given a pair of functions
which is $\delta$-slowly increasing for the 2-cycle game,
for a small enough $\delta$, the players can recover the pure Nash equilibrium. 
The proof is in Section~\ref{sec:proof-of-slowly-increasing}

\begin{lemma}\label{lem:slowly-increasing}
Consider a 2-cycle $N\times N$ game, given by the utility functions $u_A,u_B$.
Let $a:V\rightarrow[0,1]$ and $b:V\rightarrow[0,1]$ be a pair of functions 
which is $\delta$-slowly increasing for the game, 
where $\delta \in\left[0, \frac{\max\{b(V),a(V)\}}{5N^2}\right)$.
Then, there exists a deterministic communication protocol, 
that given $u_A$, $a$ and $\delta$ to player $A$, and $u_B$, $b$ and $\delta$ to player $B$,
uses $O(\log N)$ bits of communication, and at the end of the communication
player $A$ outputs an action $u\in V$ and player $B$ outputs an action $v\in V$, 
such that $(u,v)$ is the pure Nash equilibrium of the game.
\end{lemma}

We prove that an approximate correlated equilibrium for the 2-cycle game
implies the existence of a slowly increasing pair of functions.
Theorem~\ref{thm:ACEtoPure} follows from Lemma~\ref{lem:slowly-increasing} 
and Claim~\ref{clm:from-correlated-to-slowly-increasing}.

\begin{claim}\label{clm:from-correlated-to-slowly-increasing}
Let $\mu$ be an $\eps$-approximate correlated equilibrium of the 2-cycle $N\times N$ game, 
where $\eps\leq\frac{1}{24N^3}$.
Then, there exists a pair of functions $a:V\rightarrow[0,1]$ and $b:V\rightarrow[0,1]$
which is $\frac{1}{8N^3}$-slowly increasing for the game.
Moreover, player $A$ knows $a$, player $B$ knows $b$ and $b(V) \geq \frac{3}{4N}$.
\end{claim}
\begin{proof}
Define a function $a:V\rightarrow[0,1]$ as
$$ a(v) = \mu(v,N_A(v)) ~~~~\forall~ v\in V $$
and a function $b:V\rightarrow[0,1]$ as 
$$ b(v) = \mu(N_B(v),v) ~~~~\forall~ v\in V .$$
Let $v\in V$ and assume that $b(N_A(v)) \leq p$ for some $p\in\R$.
We will show that $a(v) \leq p+3\eps$.
Let $u\in N_A(v)$ (if there is no such vertex we are done).
By Definition~\ref{def:ACE}, for every $u'\in V$, 
$$ \eps \geq \mu(N_B(u'),u) - b(u) .$$
By Proposition~\ref{prop:nonzero-out-degree}, 
there exists $u'\in V$ such that $v\in N_B(u')$.
Therefore,
$$ \eps \geq \mu(v,u) - b(u) .$$
Summing over every $u\in N_A(v)$ we get that
\begin{align*}
3\eps &\geq a(v) - b(N_A(v)) \geq a(v) - p,
\end{align*}
where we bounded the left-hand side by Proposition~\ref{prop:max-in-degree} 
and the right-hand side by the assumption.
\\
The same holds when replacing $a$ with $b$, $N_A$ with $N_B$, and $\mu(N_B(u'),u)$ with $\mu(u,N_A(u'))$. 
That is, for every $v\in V$, if $a(N_B(v)) \leq p$ for some $p\in\R$,
then $b(v) \leq p + 3\eps$.
\\
Finally, we bound $b(V)$:
\begin{align*}
1 - b(V) &= \sum_{v\in V}\sum_{v'\in V:v'\neq v}\mu(N_B(v'),v) \\
&\leq \sum_{v\in V}\sum_{v'\in V:v'\neq v}\left( \mu(N_B(v),v)+\eps \right)\\
&\leq (N-1)\cdot b(V) + N^2\cdot\eps,
\end{align*}
where the first step follows from the definition of $b$ and from Proposition~\ref{prop:nonzero-out-degree}
and the second step follows from Definition~\ref{def:ACE}.
Therefore, $b(V) \geq \frac{1}{N} - N\cdot\eps \geq \frac{3}{4N}$.
\end{proof}

\subsection{On Slowly Increasing Probabilities}\label{sec:on-slowly-increasing-probabilities}

In this section we describe some useful, basic properties of slowly increasing probabilities for the 2-cycle game.

Recall that for a vertex $u\in V$, $N_A^f(u)$ is the set of vertices $v$ 
such that $(v,u)\in E_A$ but $(v,u)$ is not a back-edge.
The following proposition states 
that a back-edge adds at most $\delta$ to the probability of its vertices.

\begin{proposition}\label{prop:slowly-increasing-back-edges}
Let $\delta\in[0,1]$ and let $a:V\rightarrow[0,1]$ and $b:V\rightarrow[0,1]$ be 
a pair of functions which is $\delta$-slowly increasing for the 2-cycle game.
Let $u=(i,z)\in V$, where $i\in [n]$, $i+1\neq i^* \bmod n$ and $z\in\{0,1\}$.
Then $a(u)\leq b(N_A^f(u)) + 2\delta$.
\end{proposition}
\begin{proof}
By Proposition~\ref{prop:2-cycle}, for every back-edge $(v,u)\in E_A$, where $v \neq v^*_0$,
it holds that $d_B(v)=0$. 
Since $a(\emptyset) = 0$, by Definition~\ref{def:slowly-increasing-probabilities}, $b(v)\leq \delta$.
Therefore, 
\begin{align*}
a(u) &\leq b(N_A(u)) + \delta 
\leq  b(N_A^f(u)) + 2\delta.
\end{align*}
\end{proof}

Recall that for $i\in[n]$, $ L^{m}_{i} = \{(i,z) ~:~ z\in\{01,11,0\}\} $.
The following proposition states that
bounding the probabilities of vertices in a midway layer implies a bound on the corresponding layer.

\begin{proposition}\label{prop:slowly-increasing-midway-layers}
Let $\delta\in[0,1]$ and let $a:V\rightarrow[0,1]$ and $b:V\rightarrow[0,1]$ be 
a pair of functions which is $\delta$-slowly increasing for the 2-cycle game.
Let $i\in[n]$ such that $i+1\neq i^* \bmod n$.
Then,
$$ a(L_i) + b(L_i) \leq a(L^{m}_i) + b(L^{m}_i) + 3\delta .$$ 
\end{proposition}
\begin{proof}
By Proposition~\ref{prop:slowly-increasing-back-edges},
\begin{align*}
a(i,1) &\leq b(i,01) + b(i,11) + 2\delta
\end{align*}
and by Definition~\ref{def:slowly-increasing-probabilities},
\begin{align*}
b(i,1) &\leq a(i,01) + a(i,11) + \delta .
\end{align*}
\end{proof}

Recall that an index $i\in[n]$ is disputed if $x_i > y_i$,
where $x,y$ are the strings from which the game was constructed,
otherwise $i$ is undisputed.
The game has exactly one disputed index $i^*$.


\begin{proposition}\label{prop:slowly-increasing-undisputed}
Let $\delta\in[0,1]$ and let $a:V\rightarrow[0,1]$ and $b:V\rightarrow[0,1]$ be 
a pair of functions which is $\delta$-slowly increasing for the 2-cycle game.
Let $x,y$ be the strings from which the game was constructed.
For every $i\in[n]$, if $i$ is undisputed and $y_i = 0$ then $b(i,1) \leq 3\delta$.
\end{proposition}
\begin{proof}
Let $i\in[n]$ and assume that $i$ is undisputed and $y_i=0$.
There are exactly two edges in $G_B$ going into $(i,1)$, 
from the vertices $(i,01)$ and $(i,11)$.
Since $i$ is undisputed, $x_i=0$
and $ d_A(i,01)=d_A(i,11)=0 $.
By Definition~\ref{def:slowly-increasing-probabilities},
\begin{align*}
b(i,1) &\leq a(N_B(i,1)) + \delta \\
&\leq a(i,01) + a(i,11) + \delta \leq 3\delta .
\end{align*}
\end{proof}

\subsection{From Slowly Increasing Probabilities to The Pure Nash Equilibrium}\label{sec:proof-of-slowly-increasing}

In this section we prove Lemma~\ref{lem:slowly-increasing}.
Consider a 2-cycle $N\times N$ game, given by the utility functions $u_A,u_B$.
Let $a:V\rightarrow[0,1]$ and $b:V\rightarrow[0,1]$ be a pair of functions 
which is $\delta$-slowly increasing for the game, 
where $\delta \in\left[0, \frac{\max\{b(V),a(V)\}}{5N^2}\right)$.
By Claim~\ref{clm:pure-nash}, the pure Nash equilibrium of the game is $(u^*,v^*_0)$.

The deterministic communication protocol for finding $(u^*,v^*_0)$
is described in Algorithm~\ref{algo:slowly-increasing}.
Player $A$ gets $u_A$, $a$ and $\delta$ and player $B$ gets $u_B$, $b$ and $\delta$.
The communication complexity of this protocol is clearly at most $O(\log N)$.

\begin{algorithm}~
\begin{enumerate}

	\item \label{itm1:protocol-slowly-increasing}
		  Player $B$ checks if there exists $i\in[n]$ such that $y_i=0$ and $b(i,1) > 3\delta$.
		  If there is such an index $i$ he sends it to player $A$.
		  Then, player $A$ outputs $(i-1,x_{i-1})$ and player $B$ outputs $(i,0)$.
		  Otherwise, player $B$ sends a bit to indicate that there is no such index.
		  
	\item \label{itm2:protocol-slowly-increasing}
		  Player $A$ checks if $\delta < \nicefrac{a(V)}{5N^2}$. 
		  If it is, player $A$ finds $u=(i,z)\in V$ such that $a(u) > 5N\delta$,
		  where $i\in[n]$ and $z\in\{0,1,01,11\}$, 
		  and sends $i$ to player $B$.
	      Then, player $A$ outputs $u$ and player $B$ outputs $(i+1,0)$.
	      Otherwise, player $A$ sends a bit to indicate that $\delta \geq \nicefrac{a(V)}{5N^2}$.

	\item \label{itm3:protocol-slowly-increasing}
		  Player $B$ finds $v=(i,z)\in V$ such that $b(v) > 5N\delta$, 
		  where $i\in[n]$ and $z\in\{0,1,01,11\}$, 
		  and sends $i$ to player $A$. 
		  Then, player $A$ outputs $(i-1,x_{i-1})$ and player $B$ outputs $v$.
	      
\end{enumerate}
\caption{Finding $(u^*,v^*_0)$ given slowly increasing probabilities $(a,b)$} \label{algo:slowly-increasing}
\end{algorithm}

By Proposition~\ref{prop:slowly-increasing-undisputed}, 
if there exists $i\in[n]$ such that $y_i=0$ and $b(i,1) > 3\delta$, then $i$ has to be $i^*$.
In this case the players $A,B$ output $u^*$ and $v^*_0$ respectively.
Otherwise, $b(v^*_1) \leq 3\delta$.
In this case, the correctness of the protocol follows from Lemma~\ref{lem:slowly-increasing-concentration} below.

\begin{lemma}\label{lem:slowly-increasing-concentration}
Consider a 2-cycle $N\times N$ game, given by the utility functions $u_A,u_B$.
Let $a:V\rightarrow[0,1]$ and $b:V\rightarrow[0,1]$ be a pair of functions 
which is $\delta$-slowly increasing for the game, 
where $\delta \in\left[0, \frac{\max\{b(V),a(V)\}}{5N^2}\right)$.
Let $p = 5N\delta$.
Then, either $b(v^*_1) > 3\delta$
or the following two conditions hold:
\begin{enumerate}
\item $a$ is $p$-concentrated on $u^*$.
\item $b$ is $p$-concentrated on $v^*_0$.
\end{enumerate}
\end{lemma}

Note that if $\delta < \frac{a(V)}{5N^2}$ then $p<\frac{a(V)}{N}$ and $a(u^*) > \frac{a(V)}{N} > p$.
Otherwise, $\delta < \frac{b(V)}{5N^2}$ then $p<\frac{b(V)}{N}$ and $b(v^*_0) > \frac{b(V)}{N} > p$.
For the rest of this section we prove Lemma~\ref{lem:slowly-increasing-concentration}.
Assume that $b(v^*_1) \leq 3\delta$.
First, note that $d_A(v^*_0) = d_B(v^*_{01}) = d_B(v^*_{11}) = 0$
therefore
$$ a(v^*_0),b(v^*_{01}),b(v^*_{11}) \leq \delta .$$
and by Proposition~\ref{prop:slowly-increasing-back-edges},
\begin{align*}
a(v^*_1) 
&\leq b(N_A^f(v^*_1)) + 2\delta \\
&\leq b(v^*_{01}) + b(v^*_{11}) + 2\delta \leq 4\delta.
\end{align*}
Next, we prove that for every $1\leq j\leq n-2$,
\begin{equation}\label{eq:slowly-increasing-inductive-hypothesis}
a(L^{m}_{i^*+j}) + b(L^{m}_{i^*+j}) \leq  15j\delta.
\end{equation}
Recall that $ L^{m}_{i^*+j} = \{(i^*+j,z) ~:~ z\in\{01,11,0\}\} $.
By Proposition~\ref{prop:slowly-increasing-midway-layers},~\eqref{eq:slowly-increasing-inductive-hypothesis}
implies that $a(L_i) + b(L_i) \leq 15j\delta + 3\delta$ for every $1\leq j\leq n-2$.
Note that $15j\delta + 3\delta \leq 5N\delta$.
We prove \eqref{eq:slowly-increasing-inductive-hypothesis} by induction on $j$, from $j=1$ to $j=n-2$.

\paragraph{Layer $i^*+1$:}
First we bound $a(L^{m}_{i^*+1})$.
Note that in $E_A$ there is no edge from $v^*_0$ to $(i^*+1,0)$  or to $(i^*+1,01)$.
Moreover, by Proposition~\ref{prop:nonzero-out-degree}, 
every vertex $v\in V$ has exactly one out-going edge in each graph.
That is, in $E_A$, either there is an edge from $v^*_1$ to $(i^*+1,11)$ or 
there is an edge from $v^*_1$ to $(i^*+1,0)$, but not both.
Therefore, 
\begin{align*}
a(L^{m}_{i^*+1})
&\leq b(N_A(i^*+1,11)) + b(N_A^f(i^*+1,0)) + 4\delta \\
&\leq b(v^*_1) + 4\delta \leq 7\delta,
\end{align*}
where the first step follows from Definition~\ref{def:slowly-increasing-probabilities}, Proposition~\ref{prop:slowly-increasing-back-edges} and since $ d_A(i^*+1,01)=0 $.
\\
Next we bound $b(L^{m}_{i^*+1})$.
In $E_B$, either there are edges from $v^*_0$ and $v^*_1$ to $(i^*+1,01)$ and $(i^*+1,11)$ respectively,
or there are edges from $v^*_0$ and $v^*_1$ to $(i^*+1,0)$, but not both.
Therefore, by Definition~\ref{def:slowly-increasing-probabilities},
\begin{align*}
b(L^{m}_{i^*+1})
&\leq a(N_B(i^*+1,01)) + a(N_B(i^*+1,11)) + a(N_B(i^*+1,0)) + 3\delta \\
&\leq a(v^*_0) + a(v^*_1) + 3\delta \leq 8\delta.
\end{align*}
Put together we get that 
$ a(L^{m}_{i^*+1}) + b(L^{m}_{i^*+1}) \leq 15\delta $.

\paragraph{Layers $i^*+2, \dots, i^*+n-2$:}
Fix $i\in[n-2]$.
By Proposition~\ref{prop:nonzero-out-degree}, 
every vertex $v\in V$ has exactly one out-going edge in each graph.
That is, in each graph, 
either $(i^*+i+1,01)$ and $(i^*+i+1,11)$ have no incoming edges from layer $L_{i^*+i}$,
or $(i^*+i+1,0)$ has no incoming edges from layer $L_{i^*+i}$.
If $(i^*+i+1,01)$ and $(i^*+i+1,11)$ have no incoming edges from layer $L_{i^*+i}$ in $E_A$
then
\begin{align*}
a(L^{m}_{i^*+i+1})
&\leq a(i^*+i+1,0) + 2\delta \\
&\leq b(N_A^f(i^*+i+1,0)) + 4\delta \\
&\leq b(L_{i^*+i}) + 4\delta.
\end{align*}
where the first step holds since $d_A(i^*+i+1,01) = d_A(i^*+i+1,11) = 0$
and the second step follows from Proposition~\ref{prop:slowly-increasing-back-edges}.
Otherwise, $(i^*+i+1,0)$ has no incoming edges from layer $L_{i^*+i}$ in $E_A$
and then
\begin{align*}
a(L^{m}_{i^*+i+1})
&\leq a(i^*+i+1,01) + a(i^*+i+1,11) + 2\delta \\
&\leq b(L_{i^*+i}) + 4\delta,
\end{align*}
where the first step follows from Proposition~\ref{prop:slowly-increasing-back-edges} 
since $b(N_A^f(i^*+i+1,0)) = b(\emptyset) = 0$
and the second step follows from Definition~\ref{def:slowly-increasing-probabilities}.
\\
The same holds when replacing $a$ with $b$, and $N_A^f$ with $N_B$.
(the fact that there are no back-edges in $G_B$ could only decrease the bound).
That is,
\begin{align}
b(L^{m}_{i^*+i+1})
&\leq a(L_{i^*+i}) + 4\delta. \label{eq:slowly-increasing-step-b}
\end{align}
\\
Put together we get that
\begin{align*}
a(L^{m}_{i^*+i+1}) + b(L^{m}_{i^*+i+1})
&\leq a(L_{i^*+i}) + b(L_{i^*+i}) + 8\delta .
\end{align*}
Assume that $i < n-2$ and that~\eqref{eq:slowly-increasing-inductive-hypothesis} holds for every $1 \leq j \leq i$.
By Proposition~\ref{prop:slowly-increasing-midway-layers}, 
$a(L_{i^*+j}) + b(L_{i^*+j}) \leq 15j\delta + 3\delta$ for every $1 \leq j \leq i$.
Therefore,
\begin{align*}
a(L^{m}_{i^*+i+1}) + b(L^{m}_{i^*+i+1})
&\leq 15i\delta + 11\delta \leq 15(i+1)\delta.
\end{align*}

\paragraph{Bounding $b$ on the remaining vertices.}
It remains to bound $b$ on the vertices $(i^*-1,z)$, where $z\in\{0,1,01,11\}$.
It holds that
\begin{align*}
b(L^{m}_{i^*-1})
&\leq a(L_{i^*-2}) + 4\delta \\
&\leq a(L^m_{i^*-2}) + b(L^m_{i^*-2}) + 7\delta \\
&\leq 15(n-2)\delta + 7\delta \leq 5N\delta, 
\end{align*}
where the first step follows from~\eqref{eq:slowly-increasing-step-b},
the second step follows from Proposition~\ref{prop:slowly-increasing-midway-layers}
and the third step follows from~\eqref{eq:slowly-increasing-inductive-hypothesis}.
Finally, 
\begin{align*}
b(i^*-1,1)
&\leq a(i^*-1,01) + a(i^*-1,11) + \delta \\
&\leq b(L_{i^*-2}) + 3\delta \\
&\leq a(L^m_{i^*-2}) + b(L^m_{i^*-2}) + 6\delta \\
&\leq 15(n-2)\delta + 6\delta \leq 5N\delta, 
\end{align*}
where the first two steps follow from Definition~\ref{def:slowly-increasing-probabilities},
the third step follows from Proposition~\ref{prop:slowly-increasing-midway-layers}
and fourth step follows from~\eqref{eq:slowly-increasing-inductive-hypothesis}.

\paragraph{Bounding $a$ on the remaining vertices.}
We already bounded $a$ on the vertices $(i^*-1,01)$ and $(i^*-1,11)$.
It remains to bound $a$ on the vertices $(i^*-1,\bar{x}_{i^*-1}),(i^*,01)$ and $(i^*,11)$.
Denote $\bar{u}^* = (i^*-1,\bar{x}_{i*-1})$.
Since there is not back-edge into $\bar{u}^*$,
\begin{align*}
a(\bar{u}^*)
&\leq b(L_{i^*-2}) + b(i^*-1,01) + b(i^*-1,11) + \delta\\
&\leq b(L_{i^*-2}) + a(L_{i^*-2}) + 3\delta \\
&\leq a(L^m_{i^*-2}) + b(L^m_{i^*-2}) + 6\delta \\
&\leq 15(n-2)\delta + 6\delta \leq 5N\delta,
\end{align*}
where the first two steps follow from Definition~\ref{def:slowly-increasing-probabilities},
the third step follows from Proposition~\ref{prop:slowly-increasing-midway-layers}
and fourth step follows from~\eqref{eq:slowly-increasing-inductive-hypothesis}.
Finally,
\begin{align*}
a(i^*,01) + a(i^*,11)
&\leq b(L_{i^*-1}) + 2\delta \\
&\leq a(i^*-1,01) + a(i^*-1,11) + b(i^*-1,0) + 3\delta \\
&\leq b(L_{i^*-2}) + a(L_{i^*-2}) + 6\delta \\
&\leq a(L^m_{i^*-2}) + b(L^m_{i^*-2}) + 9\delta \\
&\leq 15(n-2)\delta + 9\delta \leq 5N\delta,
\end{align*}
where the first three steps follow from Definition~\ref{def:slowly-increasing-probabilities},
the fourth step follows from Proposition~\ref{prop:slowly-increasing-midway-layers}
and the fifth step follows from~\eqref{eq:slowly-increasing-inductive-hypothesis}.

\section{Approximate Nash Equilibrium of The 2-Cycle Game}\label{sec:ANE}

%
%
The following theorem states that 
given an approximate Nash equilibrium of the 2-cycle game, 
the players can recover the pure Nash equilibrium. 

\begin{theorem}\label{thm:ANEtoPure}
Consider a 2-cycle $N\times N$ game, given by the utility functions $u_A,u_B$.
Let $\eps\leq\frac{1}{16N^2}$ and 
let $(a^*,b^*)$ be an $\eps$-approximate Nash equilibrium of the game.
Then, there exists a deterministic communication protocol, 
that given $u_A$, $a^*$ to player $A$, and $u_B$, $b^*$ to player $B$,
uses $O(\log N)$ bits of communication, and at the end of the communication
player $A$ outputs an action $u\in V$ and player $B$ outputs an action $v\in V$, 
such that $(u,v)$ is the pure Nash equilibrium of the game.
\end{theorem}

Theorem~\ref{thm:ANE} follows from Theorem~\ref{thm:pure} and Theorem~\ref{thm:ANEtoPure}.
For the rest of this section we prove Theorem~\ref{thm:ANEtoPure}.
The proof uses the notion of non-increasing probabilities. 
For more details on non-increasing probabilities see Section~\ref{sec:on-non-increasing-probabilities}.

\begin{definition}[Non-increasing probabilities]\label{def:non-increasing-probabilities}
Let $p\in[0,1]$.
A pair of distributions $(a,b)$, each over the set of actions $V$,
is \emph{$p$-non-increasing} for the 2-cycle game
if for every $u\in V$ the following conditions hold:
\begin{enumerate}
\item If $\max_{v\in N_A(u)}b(v) \leq p$ then $a(u)\leq p$.
\item If $\max_{v\in N_B(u)}a(v) \leq p$ then $b(u)\leq p$.
\end{enumerate}
\end{definition}
\noindent In particular, if $d_A(u)=0$ then $\max_{v\in \emptyset}b(v)=0$ and therefore $a(u)\leq p$.
Similarly, if $d_B(u)=0$ then $b(u)\leq p$.

The next lemma states that given a pair of distributions 
which is $p$-non-increasing for the 2-cycle game,
for a small enough $p$, the players can recover the pure Nash equilibrium. 
The proof is in Section~\ref{sec:proof-of-non-increasing}

\begin{lemma}\label{lem:non-increasing}
Consider a 2-cycle $N\times N$ game, given by the utility functions $u_A,u_B$.
Let $(a,b)$ be a pair of distributions, each over the set of actions $V$,
which is $p$-non-increasing for the game, where $p\in\left[0,\frac{1}{N}\right)$.
Then, there exists a deterministic communication protocol, 
that given $u_A$, $a$ and $p$ to player $A$, and $u_B$, $b$ and $p$ to player $B$,
uses at most $O(\log N)$ bits of communication, and at the end of the communication
player $A$ outputs an action $u\in V$ and player $B$ outputs an action $v\in V$, 
such that $(u,v)$ is the pure Nash equilibrium of the game.
\end{lemma}

We prove that an approximate Nash equilibrium for the 2-cycle game
is a non-increasing pair of functions.
Theorem~\ref{thm:ANEtoPure} follows from Lemma~\ref{lem:non-increasing} 
and Claim~\ref{clm:from-approximate-to-non-increasing}.

\begin{claim}\label{clm:from-approximate-to-non-increasing}
Let $(a^*,b^*)$ be an $\eps$-approximate Nash equilibrium of the 2-cycle $N\times N$ game, 
where $\eps\leq\frac{1}{16N^2}$.
Then, the pair $(a^*, b^*)$ is $\frac{1}{4N}$-non-increasing for the game.
\end{claim}
\begin{proof}
Let $u'\in V$ and assume that $b^*(v)\leq\frac{1}{4N}$ for every $v\in N_A(u')$.
Let $v'\in V$ be a vertex such that $b^*(v')\geq\frac{1}{N}$ 
(there must exist such a vertex since $b^*$ is a distribution).
By Proposition~\ref{prop:nonzero-out-degree}, there exists a vertex $u''\in V$ such that $u_A(u'',v')=1$.
Note that by our assumption, $u'\neq u''$ and $u_A(u',v')=0$.
Define a distribution $a'$ over $V$ as follows: 
\begin{align*}
& a'(u'') = a^*(u'') + a^*(u') \\
& a'(u')=0 \\
& a'(u)=a^*(u) ~~~~ \forall~ u\in V\setminus\{u'',u'\}.
\end{align*}
By Definition~\ref{def:ANE},
\begin{align*}
\eps &\geq \Ex_{u\sim a', v\sim b^*}[u_A(u,v)] - \Ex_{u\sim a^*, v\sim b^*}[u_A(u,v)] \\
&= a'(u'')\cdot b^*(N_A(u'')) 
	- a^*(u'')\cdot b^*(N_A(u''))
	- a^*(u')\cdot b^*(N_A(u')) \\
&= a^*(u')\cdot b^*(N_A(u''))
	- a^*(u')\cdot b^*(N_A(u')) \\
&\geq a^*(u')\cdot\frac{1}{N} - a^*(u')\cdot \frac{3}{4N} ,
\end{align*}
where the last step follows from Proposition~\ref{prop:max-in-degree}.
Since $\eps\leq\frac{1}{16N^2}$ we conclude that $a^*(u')\leq\frac{1}{4N}$. 
\\
The same holds when replacing $a^*$ with $b^*$, $N_A$ with $N_B$, and $u_A$ with $u_B$. 
That is, for every $u'\in V$, if $a^*(v)\leq\frac{1}{4N}$ for every $v\in N_B(u')$
then $b^*(u')\leq\frac{1}{4N}$.
\end{proof}

%
%
%

\subsection{On Non-Increasing Probabilities}\label{sec:on-non-increasing-probabilities}

In this section we describe some useful, basic properties of non-increasing probabilities for the 2-cycle game.

Recall that for a vertex $u\in V$, $N_A^f(u)$ is the set of vertices $v$ 
such that $(v,u)\in E_A$ but $(v,u)$ is not a back-edge.
The following proposition states that 
back-edges can be ignored when bounding the probabilities of out-neighbors.

\begin{proposition}\label{prop:non-increasing-back-edges}
Let $p\in[0,1]$ and let $(a,b)$ be a pair of distributions, each over the set of actions $V$,
which is $p$-non-increasing for the 2-cycle game.
Let $(v,u)\in E_A$ be a back-edge, where $v\neq v^*_0$.
Assume that $\max_{v\in N_A^f(u)}b(v) \leq p$, then $a(u)\leq p$.
\end{proposition}
\begin{proof}
By Proposition~\ref{prop:2-cycle}, for every back-edge $(v,u)\in E_A$, where $v \neq v^*_0$,
it holds that $d_B(v)=0$ and therefore $b(v)\leq p$.
We get that
$$ \max_{v\in N_A^f(u)}b(v) \leq p ~~\Rightarrow~ \max_{v\in N_A(u)}b(v) \leq p .$$
\end{proof}

Recall that for $i\in[n]$, $ L^{m}_{i} = \{(i,z) ~:~ z\in\{01,11,0\}\} $.
The following proposition states that 
bounding the probabilities of vertices in a midway layer implies a bound on the corresponding layer.

\begin{proposition}\label{prop:non-increasing-midway-layers}
Let $p\in[0,1]$ and let $(a,b)$ be a pair of distributions, each over the set of actions $V$,
which is $p$-non-increasing for the 2-cycle game.
Let $i\in[n]$ such that $i+1\neq i^* \bmod n$. 
If $a(v),b(v) \leq p$ for every $v\in L^m_{i}$ then $a(i,1),b(i,1) \leq p$.
\end{proposition}
\begin{proof}
Assume that $a(v),b(v) \leq p$ for every $v\in L^m_{i}$. Then,
\begin{align*}
\max_{v\in N_B(i,1)} \{a(v)\}
&= \max\{ a(i,01),a(i,11) \} \leq p.
\end{align*}
By Definition~\ref{def:non-increasing-probabilities}, $b(i,1) \leq p$ 
and similarly,
\begin{align*}
\max_{v\in N_A^f(i,1)} \{b(v)\}
&= \max\{ b(i,01),b(i,11) \} \leq p.
\end{align*}
By Proposition~\ref{prop:non-increasing-back-edges}, $a(i,1) \leq p$.
\end{proof}

Recall that an index $i\in[n]$ is disputed if $x_i > y_i$,
where $x,y$ are the strings from which the game was constructed,
otherwise $i$ is undisputed.
The game has exactly one disputed index $i^*$.


\begin{proposition}\label{prop:non-increasing-undisputed}
Let $p\in[0,1]$ and let $(a,b)$ be a pair of distributions, each over the set of actions $V$,
which is $p$-non-increasing for the 2-cycle game.
Let $x,y$ be the strings from which the game was constructed.
For every $i\in[n]$, if $i$ is undisputed and $y_i = 0$ then $b(i,1) \leq p$.
\end{proposition}
\begin{proof}
Let $i\in[n]$ and assume that $i$ is undisputed and $y_i=0$.
There are exactly two edges in $G_B$ going into $(i,1)$, 
from the vertices $(i,01)$ and $(i,11)$.
Since $i$ is undisputed, $x_i=0$
and $ d_A(i,01)=d_A(i,11)=0 $.
Therefore, $b(i,1)\leq p$.
\end{proof}

\subsection{From Non-Increasing Probabilities to The Pure Nash Equilibrium}\label{sec:proof-of-non-increasing}

In this section we prove Lemma~\ref{lem:non-increasing}.
Consider a 2-cycle $N\times N$ game, given by the utility functions $u_A,u_B$.
Let $(a,b)$ be a pair of distributions, each over the set of actions $V$, 
which is $p$-non-increasing for the game, where $p\in\left[0,\frac{1}{N}\right)$.
By Claim~\ref{clm:pure-nash}, the pure Nash equilibrium of the game is $(u^*,v^*_0)$.

The deterministic communication protocol for finding $(u^*,v^*_0)$
is described in Algorithm~\ref{algo:non-increasing}.
Player $A$ gets $u_A$, $a$ and $p$ and player $B$ gets $u_B$, $b$ and $p$.
The communication complexity of this protocol is clearly at most $O(\log N)$.

\begin{algorithm}~

	Player $B$ checks if there exists $i\in[n]$ such that $y_i=0$ and $b(i,1) > p$.
	If there is such an index $i$, player $B$ sends $i$ to player $A$.
	Then, player $A$ outputs $(i-1,x_{i-1})$ and player $B$ outputs $(i,0)$.
	Otherwise, player $B$ sends a bit to indicate that there is no such index.
	Then, player $A$ outputs $u\in V$ such that $a(u)>p$
	and player $B$ outputs $v\in V$ such that $b(v)>p$.

\caption{Finding $(u^*,v^*_0)$ given non-increasing probabilities $(a,b)$} \label{algo:non-increasing}
\end{algorithm}

By Proposition~\ref{prop:non-increasing-undisputed}, 
if there exists $i\in[n]$ such that $y_i=0$ and $b(i,1) > p$, then $i$ has to be $i^*$.
In this case the players $A,B$ output $u^*$ and $v^*_0$ respectively.
Otherwise, $b(v^*_1) \leq p$.
In this case, the correctness of the protocol follows from Lemma~\ref{lem:non-increasing-concentration} below.
Note that since $p < \frac{1}{N}$, we have that 
$a(u^*) > \frac{1}{N} > p$ and $b(v^*_0) > \frac{1}{N} > p$.

\begin{lemma}\label{lem:non-increasing-concentration}
Consider a 2-cycle $N\times N$ game, given by the utility functions $u_A,u_B$.
Let $(a,b)$ be a pair of distributions, each over the set of actions $V$, 
which is $p$-non-increasing for the game, where $p\in\left[0,\frac{1}{N}\right)$.
Then, either $b(v^*_1) > p$
or the following two conditions hold:
\begin{enumerate}
\item $a$ is $p$-concentrated on $u^*$.
\item $b$ is $p$-concentrated on $v^*_0$.
\end{enumerate}
\end{lemma}

\begin{remark}
The proof of Lemma~\ref{lem:slowly-increasing-concentration} in Section~\ref{sec:proof-of-slowly-increasing}
is slightly more delicate than the proof of 
Lemma~\ref{lem:non-increasing-concentration}.
Unlike the analysis of the slowly-increasing probabilities, here we do not use the fact that  
every vertex has exactly one out-going edge in each graph (see Proposition~\ref{prop:nonzero-out-degree}),
as this would not improve the parameters of Theorem~\ref{thm:ANE}.
\end{remark}

For the rest of this section we prove Lemma~\ref{lem:non-increasing-concentration}.
Assume that $b(v^*_1) \leq p$.
First, note that $d_A(v^*_0) = d_B(v^*_{01}) = d_B(v^*_{11}) = 0$
therefore 
$$ a(v^*_0),b(v^*_{01}),b(v^*_{11}) \leq p .$$
By Proposition~\ref{prop:non-increasing-back-edges}, $a(v^*_1) \leq p$ since
\begin{align*}
\max_{v\in N_A^f(v^*_1)}b(v)
&\leq \max\{ b(v^*_{01}),b(v^*_{11})\} \leq p.
\end{align*}
Next, we prove that for every $1\leq j\leq n-2$ and every $z\in\{0,01,11\}$,
\begin{equation}\label{eq:non-increasing-inductive-hypothesis} 
a(i^*+j,z), b(i^*+j,z) \leq p .
\end{equation}
By Proposition~\ref{prop:non-increasing-midway-layers}, 
\eqref{eq:non-increasing-inductive-hypothesis} implies that $a(i^*+j,1),b(i^*+j,1) \leq p$ for every $1\leq j\leq n-2$.
We prove \eqref{eq:non-increasing-inductive-hypothesis} by induction on $j$, from $j=1$ to $j=n-2$.

\paragraph{Layer $i^*+1$:}
Since the vertices $(i^*+1,01)$, $(i^*+1,11)$ and $(i^*+1,0)$ have no incoming edges from $v^*_0$ in $E_A$,
\begin{align*}
\max_{v\in N_A^f\left(L^{m}_{i^*+1}\right)} \{b(v)\} 
&\leq b(v^*_1) \leq p.
\end{align*}
Therefore, by Proposition~\ref{prop:non-increasing-back-edges}, $a(v)\leq p$ for every $v\in L^m_{i^*+1}$. 
Next, 
\begin{align*} 
\max_{v\in N_B\left( L^{m}_{i^*+1} \right)} \{a(v)\}
&\leq \max\{ a(v^*_0), a(v^*_1) \} 
\leq p. 
\end{align*}
Therefore, by Definition~\ref{def:non-increasing-probabilities}, $b(v)\leq p$ for every $v\in L^m_{i^*+1}$. 

\paragraph{Layers $i^*+2, \dots, i^*+n-2$:}
Fix $i\in[n-3]$ 
and assume that~\eqref{eq:non-increasing-inductive-hypothesis} holds for every $1 \leq j \leq i$.
By Proposition~\ref{prop:non-increasing-midway-layers}, $a(i^*+j,1),b(i^*+j,1) \leq p$ for every $1 \leq j \leq i$.
We get that 
\begin{align*}
\max_{v\in N_A^f\left( L^{m}_{i^*+i+1} \right)}\{b(v)\}
&\leq \max\{ b(i^*+i,0),b(i^*+i,1) \} \leq p.
\end{align*}
By Proposition~\ref{prop:non-increasing-back-edges}, $a(v)\leq p$ for every $v\in L^m_{i^*+i+1}$.
\\
The same holds when replacing $a$ with $b$, and $N_A^f$ with $N_B$
(the fact that there are back-edges in $G_A$ but not in $G_B$ does not change the bounds).
That is, $ b(v)\leq p$ for every $v\in L^m_{i^*+i+1}$.

\paragraph{Bounding $b$ on the remaining vertices.}
It remains to bound $b$ on the vertices $(i^*-1,z)$, where $z\in\{0,1,01,11\}$.
It holds that
\begin{align*}
\max_{v\in N_B\left( L^{m}_{i^*-1} \right)}\{a(v)\}
&\leq \max\{ a(i^*-2,0),a(i^*-2,1) \} \leq p,
\end{align*}
where the last step follows from~\eqref{eq:non-increasing-inductive-hypothesis} 
and Proposition~\ref{prop:non-increasing-midway-layers}.
Therefore, by Definition~\ref{def:non-increasing-probabilities}, 
\begin{equation}\label{eq:b-on-n-1-midway}
b(v)\leq p ~~~~\forall~ v\in L^m_{i^*-1}.
\end{equation}
Finally, 
\begin{align*}
\max_{v\in N_A(i^*-1,01)\cup N_A(i^*-1,11)} \{b(v)\}
&\leq \max\{ b(i^*-2,0),b(i^*-2,1) \} \leq p,
\end{align*}
where the last step follows from~\eqref{eq:non-increasing-inductive-hypothesis} 
and Proposition~\ref{prop:non-increasing-midway-layers}.
Therefore, by Definition~\ref{def:non-increasing-probabilities}, 
$ a(i^*-1,01),a(i^*-1,11) \leq p $ and
\begin{align*}
\max_{v\in N_B(i^*-1,1)}\{a(v)\}
&\leq \max\{ a(i^*-1,01),a(i^*-1,11) \} \leq p.
\end{align*}
Therefore, by Definition~\ref{def:non-increasing-probabilities}, 
\begin{equation}\label{eq:b-on-n-1-1}
b(i^*-1,1)\leq p .
\end{equation}

\paragraph{Bounding $a$ on the remaining vertices.}
We already bounded $a$ on the vertices $(i^*-1,01)$ and $(i^*-1,11)$.
It remains to bound $a$ on the vertices $(i^*-1,\bar{x}_{i^*-1}),(i^*,01)$ and $(i^*,11)$.
Denote $\bar{u}^* = (i^*-1,\bar{x}_{i*-1})$.
Since there is not back-edge into $\bar{u}^*$,
\begin{align*}
\max_{v\in N_A(\bar{u}^*)  }\{b(v)\}
&\leq \max\{b(i^*-2,0),b(i^*-2,1),b(i^*-1,01),b(i^*-1,11)\} \leq p,
\end{align*}
where the last step follows from~\eqref{eq:non-increasing-inductive-hypothesis},
Proposition~\ref{prop:non-increasing-midway-layers} and~\eqref{eq:b-on-n-1-midway}.
Therefore, by Definition~\ref{def:non-increasing-probabilities}, $a(\bar{u}^*) \leq p$.
Finally,
\begin{align*}
\max_{v\in N_A(i^*,01)\cup N_A(i^*,11)}\{b(v)\}
&\leq \max\{b(i^*-1,0),b(i^*-1,1)\} \leq p,
\end{align*}
where the last step follows from~\eqref{eq:b-on-n-1-midway} and~\eqref{eq:b-on-n-1-1}.
Therefore, by Definition~\ref{def:non-increasing-probabilities}, $a(i^*,01),a(i^*,11) \leq p$.

\subsection{Approximate Well Supported Nash Equilibrium}\label{sec:AWSNE}

Since every $\eps$-approximate well supported Nash equilibrium is an $\eps$-approximate Nash equilibrium (see Proposition~\ref{WSNEtoNE}),
Theorem~\ref{thm:ANE} gives a lower bound for the communication complexity of finding 
$\eps$-approximate well supported Nash equilibrium, for $\eps\leq\frac{1}{16N^2}$.
However, the following claim shows that every $\eps$-approximate well supported Nash equilibrium, 
for $\eps\leq\frac{1}{N}$, is a pair of $0$-non-increasing functions for the 2-cycle game.
Theorem~\ref{thm:AWSNE} follows from Lemma~\ref{lem:non-increasing} 
and Claim~\ref{clm:from-awsne-to-non-increasing}.

\begin{claim}\label{clm:from-awsne-to-non-increasing}
Let $(a^*,b^*)$ be an $\eps$-approximate well supported Nash equilibrium 
of the 2-cycle $N\times N$ game, where $\eps\leq\frac{1}{N}$. 
Then, the pair $(a^*, b^*)$ is $0$-non-increasing for the game.
\end{claim}
\begin{proof}
Let $u'\in V$ and assume that $b^*(v)=0$ for every $v\in N_A(u')$.
Let $v'\in V$ be a vertex such that $b^*(v')\geq\frac{1}{N}$ 
(there must exist such a vertex since $b^*$ is a distribution).
By Proposition~\ref{prop:nonzero-out-degree}, there exists a vertex $u''\in V$ such that $u_A(u'',v')=1$.
Note that by our assumption, $u_A(u',v')=0$.
By Definition~\ref{def:AWSNE},
\begin{align*}
\Ex_{v\sim b^*} [ u_A(u'',v) - u_A(u',v) ] 
&\geq b^*(v')u_A(u'',v')
\geq \frac{1}{N}.
\end{align*}
Therefore, $u'\notin \Supp(a^*)$.
\\
The same holds when replacing $a^*$ with $b^*$, $N_A$ with $N_B$, and $u_A$ with $u_B$. 
That is, 
for every $u'\in V$, if $a^*(v)=0$ for every $v\in N_B(u')$,
then $u'\notin \Supp(b^*)$.
\end{proof}



\subsection{Approximate Bayesian Nash Equilibrium}\label{sec:ABNE}

We prove a lower bound for finding 
an approximate Bayesian Nash equilibrium of a game called the Bayesian 2-cycle game.
The Bayesian 2-cycle game is constructed from sub-games, 
where each sub-game is similar to the 2-cycle game.
We use the construction defined in Section~\ref{sec:the-game} on strings
that have \emph{at most} one disputed index. We call the resulted game the \emph{no-promise 2-cycle game}.
%

\subsubsection*{The Bayesian 2-Cycle Game}

Let $x,y$ be two $k$-bit strings, where $k = T\cdot n$ for some $T\geq 2$ and $n\geq 3$.
Assume there exists exactly one index $i\in[k]$, such that $x_{i} > y_{i}$.

\paragraph{The graphs.}
For every $i\in[T]$ let $G_A^i = G_A^i(V,E_A^i)$ be the graph constructed by player $A$ 
from the $n$-bit string $x^i = x_{n\cdot(i-1) + 1}x_{n\cdot(i-1) + 2} \dots x_{n\cdot(i-1) + n}$, 
as defined in Section~\ref{sec:the-game}.
Similarly, for every $i\in[T]$ let $G_B^i = G_B^i(V,E_B^i)$ be the graph constructed by player $B$ 
from the $n$-bit string $y^i = y_{n\cdot(i-1) + 1}y_{n\cdot(i-1) + 2} \dots y_{n\cdot(i-1) + n}$, 
as defined in Section~\ref{sec:the-game}.
Note that $x = x^1x^2\dots x^t$ and $y=y^1y^2\dots y^t$.

\paragraph{The actions, types, prior distribution and utility functions.}
Define $\Theta_A = \Theta_B = [T]$, $\Sigma_A = \Sigma_B = V$
and $\phi$ is set to be the uniform distribution over the set $\{(i,i) ~:~ i\in[T] \}$.
For $i\in[T]$, let $u_A^i,u_B^i$ be the utility functions 
associated with the graphs $G_A^i$ and $G_B^i$ respectively, as defined in Section~\ref{sec:the-game}.
Note that the utility functions $u_A^i$ and $u_B^i$ define a no-promise 2-cycle $N\times N$ game, where $N=4n$.
We call it the $i^{\text{th}}$ sub-game.
The utility function $u_A:[T]\times V\times V\rightarrow\bits$ of player $A$ is defined 
for every type $i\in[T]$ and every pair of actions $(u,v)\in V^2$ as 
$$ u_A(i,u,v)= u_A^i(u,v) = \begin{cases*}1 &if $(v,u)\in E_A^i$ \\0 &otherwise\end{cases*} .$$
The utility function $u_B:[T]\times V\times V\rightarrow\bits$ of player $B$ is defined 
for every type $i\in[T]$ and every pair of actions $(u,v)\in V^2$ as 
$$ u_B(i,u,v)= u_B^i(u,v) = \begin{cases*}1 &if $(u,v)\in E_B^i$ \\0 &otherwise\end{cases*} .$$
This is a Bayesian game on $N=4n$ actions and $T$ types.

\subsubsection*{Pure Nash Equilibrium of a Sub-Game}

Let $i\in[T]$.
If the $i^{\text{th}}$ sub-game has a disputed index $j\in[n]$ (that is, $x^i_j > y^i_j$),
then it is a 2-cycle game and by Claim~\ref{clm:pure-nash} 
it has exactly one pure Nash equilibrium $(u^i,v^i)$, 
where $u^i = (j-1,x^i_{j-1})$ and $v^i = (j,y^i_j)$.
If the $i^{\text{th}}$ sub-game has no disputed index, then it has no pure Nash equilibrium.

Let $x,y$ be the $k$-bit strings from which the Bayesian game was constructed.
Note that $u_A$ determines $x$, and $u_B$ determines $y$.
Since there exists exactly one index $i^*\in[k]$ for which $x_{i^*} > y_{i^*}$,
there exists exactly one type $i\in[T]$ such that 
the $i^{\text{th}}$ sub-game has a pure Nash equilibrium $(u^i,v^i)$, 
where $u^i = (j-1,x^i_{j-1})$, $v^i = (j,y^i_j)$ and $i^* = n\cdot(i-1)+j$.

If player $A$ knows $u^i\in V$ and player $B$ knows $v^i\in V$ such that 
$(u^i,v^i)$ is a pure Nash equilibrium of the $i^{\text{th}}$ sub-game, for some $i\in[T]$,
then both players know the index $i^*\in[k]$ for which $x_{i^*} > y_{i^*}$.
Therefore, finding a pure Nash equilibrium of a sub-game is hard.

\begin{claim}\label{clm:pure-of-sub-game}
Every randomized communication protocol 
for finding a pure Nash equilibrium of a sub-game
of the Bayesian 2-cycle game on $N$ actions and $T$ types,
with error probability at most $\frac{1}{3}$, 
has communication complexity at least $\Omega(N\cdot T)$.
\end{claim}

Let $\eps \ge 0$ and let $\{a^*_i,b^*_i\}_{i\in[T]}$ be an $\eps$-approximate Bayesian Nash equilibrium 
of the Bayesian 2-cycle game on $N$ actions and $T$ types.
Since the prior distribution $\phi$ is uniform over the set $\{(i,i) ~:~ i\in[T] \}$,
for every $i\in[T]$, $(a^*_i,b^*_i)$ is an $\eps$-approximate Nash equilibrium
of the $i^{\text{th}}$ sub-game.
The following claim states that, for every $i\in[T]$,
if the $i^{\text{th}}$ sub-game has no pure Nash equilibrium
then $a^*_i$ and $b^*_i$ cannot be concentrated.

\begin{claim}\label{clm:no-pure-nash-no-concentration}
Let $p=\frac{1}{4N}$ and let $\eps \leq \frac{p}{4}$. 
and let $(a,b)$ be an $\eps$-approximate Nash equilibrium of the no-promise 2-cycle $N\times N$ game.
Assume that the game has no pure Nash equilibrium.
Then, for every $v$, $a$ is not $p$-concentrated on $v$ and $b$ is not $p$-concentrated on $v$.
\end{claim}
\begin{proof}
Let $u'\in V$. 
We prove that $a$ is not $p$-concentrated on $u'$.
The proof that $b$ is not $p$-concentrated on $u'$ is similar.
Assume towards a contradiction that $a$ is $p$-concentrated on $u'$.
By Proposition~\ref{prop:nonzero-out-degree}, there exists $v'\in V$ such that $(u',v')\in E_B$. 
Let $v''\in V$, $v''\neq v'$. Note that $u'\notin N_B(v'')$.
Define a distribution $b'$ over $V$ as follows: 
\begin{align*}
& b'(v') = b(v') + b(v'') \\
& b'(v'')=0 \\
& b'(v)=b(v) ~~~~ \forall~ v\in V\setminus\{v'',v'\}.
\end{align*}
It holds that 
\begin{align*}
\eps &\geq \Ex_{u\sim a, v\sim b'}[u_B(u,v)] - \Ex_{u\sim a, v\sim b}[u_B(u,v)] \\
&= b'(v')\cdot a(N_B(v')) 
	- b(v')\cdot a(N_B(v'))
	- b(v'')\cdot a(N_B(v'')) \\
&= b(v'')\cdot a(N_B(v'))
	- b(v'')\cdot a(N_B(v'')) \\
&\geq b(v'')\cdot(1-(N-1)p) - b(v'')\cdot 3p ,
\end{align*}
where the last step follows from Proposition~\ref{prop:max-in-degree}.
Since $p=\frac{1}{4N}$ we conclude that $b(v'')\leq 4\eps \leq p$. 
That is, $b$ is $p$-concentrated on $v'$.
\\
Note that Proposition~\ref{prop:nonzero-out-degree} holds also for no-promise 2-cycle games.
Therefore, there exists $u''\in V$ such that $(v',u'')\in E_A$. 
Repeating the same argument with $b$ instead of $a$, $u''$ instead of $v'$ and $N_A$ instead of $N_B$, 
we get that $a$ is $p$-concentrated on $u''$.
Therefore, $u'' = u'$, but that can only happen if $(u',v)$ is a 2-cycle,
which is a contradiction.
\end{proof}

\subsubsection*{Our Lower Bound}

In this section we prove Theorem~\ref{thm:ABNE}.
%
%
Let $\eps\leq\frac{1}{16N^2}$ and assume towards a contradiction 
that there is a randomized communication protocol $\pi$ 
that finds an $\eps$-approximate Bayesian Nash equilibrium 
of the Bayesian 2-cycle game on $N$ actions and $T$ types, 
with error probability at most $\frac{1}{3}$ and communication $o(k)$,
where $k = N\cdot T$.

Given utility functions $u_A,u_B$ of the Bayesian 2-cycle game to players $A$ and $B$ respectively,
the players run $\pi$ on the utility functions, exchanging $o(k)$ bits, and 
at the end of the communication, with probability at least $\frac{2}{3}$, 
player $A$ has a set of mixed strategies $\{a^*_i\}_{i\in[T]}$ 
and player $B$ has a set of mixed strategies $\{b^*_i\}_{i\in[T]}$, 
such that $\{a^*_i,b^*_i\}_{i\in[T]}$ is an $\eps$-approximate Bayesian Nash equilibrium of the game.

For $p\in[0,1]$, we define $p$-non-increasing for the no-promise 2-cycle game the same way that 
$p$-non-increasing for the 2-cycle game are defined.
Let $p = \frac{1}{4N}$ and let $i\in[T]$. 
Note that Claim~\ref{clm:from-approximate-to-non-increasing} holds also when 
replacing the 2-cycle game with a no-promise 2-cycle game.
Therefore, since $(a^*_i,b^*_i)$ is an $\eps$-approximate Nash equilibrium for the $i^{\text{th}}$ sub-game,
the pair $(a^*_i,b^*_i)$ is $p$-non-increasing for the $i^{\text{th}}$ sub-game.

If the $i^{\text{th}}$ sub-game has a pure Nash equilibrium
then it is a 2-cycle game and there exists an index $j\in[n]$ for which $x^i_j > y^i_j$.
By Lemma~\ref{lem:non-increasing-concentration}, 
either $b^*_i(j,1) > p$ or $b^*_i$ is $p$-concentrated on $(j,0)$.
Otherwise, the $i^{\text{th}}$ sub-game has no pure Nash equilibrium.
Note that Proposition~\ref{prop:non-increasing-undisputed} holds also when 
replacing the 2-cycle game with a no-promise 2-cycle game.
That is, for every $j\in[n]$ such that $y^i_j=0$, since $x^i_j \leq y^i_j$, it holds that $b^*_i(j,1) \leq p$.
Moreover, by Claim~\ref{clm:no-pure-nash-no-concentration}, $b^*_i$ is not $p$-concentrated on $(j,0)$.
Therefore, player $B$ can determine if the $i^{\text{th}}$ sub-game has a pure Nash equilibrium or not.

Let $i\in[T]$ be the type for which the $i^{\text{th}}$ sub-game has a pure Nash equilibrium.
That is, the $i^{\text{th}}$ sub-game is a 2-cycle game.
Player $B$ finds $i$ and sends it to player $A$, using $\log T$ bits of communication.
Then, the players run the protocol guaranteed by Lemma~\ref{lem:non-increasing}, 
for finding the pure Nash equilibrium of the $i^{\text{th}}$ sub-game, exchanging at most $O(\log N)$ bits.
That is, the players can find the pure Nash equilibrium of a sub-game 
with communication $o(k)$ and error probability at most $\frac{1}{3}$, 
which is a contradiction to Claim~\ref{clm:pure-of-sub-game}. 

\section{Open Problems}\label{sec:open}
\input{OpenProblems}

\addcontentsline{toc}{section}{\protect\numberline{}Acknowledgements}%
\section*{Acknowledgements}
We thank Amir Shpilka, Noam Nisan and Aviad Rubinstein for very helpful conversations.

\addcontentsline{toc}{section}{\protect\numberline{}References}%
\bibliographystyle{alpha}
\bibliography{refs}

\clearpage

\appendix

\section{Trivial Approximate Equilibria of The 2-Cycle Game}\label{app:examples}
\input{Appendix}

\end{document}

%% file: abstract.tex
We show a communication complexity lower bound for finding a correlated equilibrium of a two-player game.
More precisely, we define a two-player $N\times N$ game called the 2-cycle game
and show that the randomized communication complexity 
of finding a $\nicefrac{1}{\poly(N)}$-approximate correlated equilibrium 
of the 2-cycle game is $\Omega(N)$.
For small approximation values, this answers an open question of Babichenko and Rubinstein (STOC 2017).
%
%
Our lower bound is obtained via a direct reduction from the unique set disjointness problem.

%% file: Introduction.tex
\setlength{\epigraphwidth}{0.65\textwidth}
\epigraph{\textsf{If there is intelligent life on other planets, in a majority of
them, they\\ would have discovered correlated equilibrium
before Nash equilibrium.
}}{\textit{Roger Myerson}}

One of the most famous solution concepts in game theory is Nash equilibrium \cite{N51}.
Roughly speaking, a Nash equilibrium is a set of mixed strategies, one per player, 
from which no player has an incentive to deviate.
A well-studied computational problem in algorithmic game theory is that of
finding a Nash equilibrium of a given (non-cooperative) game.
The complexity of finding a Nash equilibrium has been studied in several models of computation, 
including computational complexity, query complexity and communication complexity.
Since finding a Nash equilibrium is considered a hard problem,
researchers studied the problem of finding an approximate Nash equilibrium,
where intuitively, no player can benefit much by deviating from his mixed strategy. 
For surveys on algorithmic game theory in general and equilibria in particular see for example \cite{N07, R10, G11, Rou16}.

%


A natural setting in which approximate equilibria concepts are studied 
is the setting of uncoupled dynamics \cite{HMC03, HMC06},
where each player knows his own utilities and not those of the other players.
The rate of convergence of uncoupled dynamics to an approximate equilibrium 
is closely related to the communication complexity 
of finding the approximate equilibrium \cite{CS04}.

Communication complexity is a central model in complexity theory that has been extensively studied. 
In the two-player randomized model \cite{Y79}
each player gets an input and their goal is to solve a communication task that depends on both inputs. 
The players can use both common and private random coins 
and are allowed to err with some small probability. 
The communication complexity of a protocol is the total number of bits communicated by the two players. 
The communication complexity of a communication task is the minimal number
of bits that the players need to communicate in order to solve the task with high probability,
where the minimum is taken over all protocols. 
For surveys on communication complexity see for example \cite{KN97, LS09, Rou16a}.

In a recent breakthrough, Babichenko and Rubinstein \cite{BR16} proved the first non-trivial lower bound 
on the randomized communication complexity of finding an approximate Nash equilibrium.

An important generalization of Nash equilibrium is correlated equilibrium \cite{A74, A87}.
Whereas in a Nash equilibrium the players choose their strategies independently, 
in a correlated equilibrium the players can coordinate their decisions, choosing a joint strategy.
Babichenko and Rubinstein \cite{BR16} raised the following questions:
\vspace{3mm}
\begin{center}
\begin{minipage}{.85\textwidth}
\emph{Does a $\polylog(N)$ communication protocol 
for finding an approximate correlated equilibrium of two-player $N\times N$ games exist?\\
Is there a $\poly(N)$ communication complexity lower bound?}
\end{minipage}
\end{center}
\vspace{3mm}

\noindent We answer these questions for small approximation values.
As far as we know, prior to this work, no non-trivial answers were known (neither positive nor negative),
not even for the problem of finding an \emph{exact} correlated equilibrium of two-player games.
In contrast, in the multi-party setting, there is a protocol for finding an exact correlated equilibrium of $n$-player binary action games with $\poly(n)$ bits of communication \cite{HM10, PR08, JL15}.

There are two notions of correlated equilibrium 
which we call \emph{correlated} and \emph{rule correlated} equilibria.
In a correlated equilibrium no player can benefit from replacing one action with another,
whereas in a rule correlated equilibrium no player can benefit from 
simultaneously replacing every action with another action (using a switching rule).
While the above two notions are equivalent,
approximate correlated and approximate rule correlated equilibria are not equivalent, but are closely related. 

Our first communication complexity lower bound is
for finding a $\nicefrac{1}{\poly(N)}$-approximate correlated equilibrium 
of a two-player $N\times N$ game called the 2-cycle game.
We note that every two-player $N\times N$ game
has a trivial $\nicefrac{1}{N}$-approximate correlated equilibrium (which can be found with zero communication).

\begin{theorem}\label{thm:ACE}
For every $\eps\leq\frac{1}{24N^3}$, every randomized communication protocol 
for finding an $\eps$-approximate correlated equilibrium 
of the 2-cycle $N\times N$ game, 
with error probability at most $\frac{1}{3}$, 
has communication complexity at least $\Omega(N)$.
\end{theorem}


Since every approximate rule correlated equilibrium is an approximate correlated equilibrium, 
the following lower bound follows from Theorem~\ref{thm:ACE}.
It remains an interesting open problem to prove bounds on the communication complexity of finding 
a \emph{constant} approximate rule correlated equilibrium of two-player games.

\begin{theorem}\label{thm:ARCE}
For every $\eps\leq\frac{1}{24N^3}$, every randomized communication protocol 
for finding an $\eps$-approximate rule correlated equilibrium 
of the 2-cycle $N\times N$ game, 
with error probability at most $\frac{1}{3}$, 
has communication complexity at least $\Omega(N)$.
\end{theorem}

Note that Theorems~\ref{thm:ACE} and~\ref{thm:ARCE}
imply a lower bound of $\Omega(N)$
for the randomized \emph{query} complexity of finding 
a $\nicefrac{1}{\poly(N)}$-approximate correlated, respectively rule correlated, equilibrium
of the 2-cycle game on $N\times N$ actions.

Next, we show a communication complexity lower bound 
for finding a $\nicefrac{1}{\poly(N)}$-approximate Nash equilibrium of the 2-cycle game.
As previously mentioned, Babichenko and Rubinstein \cite{BR16} 
proved the first non-trivial lower bound on the randomized communication complexity 
of finding an approximate Nash equilibrium.
More precisely, they proved a lower bound of $\Omega\left(N^{\eps_0}\right)$ 
on the randomized communication complexity of finding an $\eps$-approximate Nash equilibrium 
of a two-player $N\times N$ game,
for every $\eps \le \eps_0$, where $\eps_0$ is some small constant. 
Their proof goes through few intermediate problems and involves intricate reductions.
We believe our proof is more simple and straightforward.
Moreover, for small approximation values, 
we get a stronger lower bound of $\Omega(N)$, as opposed to the $\Omega(N^{\eps_0})$ lower bound of \cite{BR16}.

\begin{theorem}\label{thm:ANE}
For every $\eps\leq\frac{1}{16N^2}$, every randomized communication protocol 
for finding an $\eps$-approximate Nash equilibrium
of the 2-cycle $N\times N$ game, 
with error probability at most $\frac{1}{3}$, 
has communication complexity at least $\Omega(N)$.
\end{theorem}

Using similar ideas to the ones used in the proof of Theorem~\ref{thm:ANE}, 
we get a communication complexity lower bound for finding 
an approximate well supported Nash equilibrium of the 2-cycle game.

\begin{theorem}\label{thm:AWSNE}
For every $\eps\leq\frac{1}{N}$, every randomized communication protocol 
for finding an $\eps$-approximate well supported Nash equilibrium 
of the 2-cycle $N\times N$ game, 
with error probability at most $\frac{1}{3}$, 
has communication complexity at least $\Omega(N)$.
\end{theorem}

The 2-cycle game is a very simple game, in the sense that it is 
a win-lose, sparse game, in which each player has a unique best response to every action.
Moreover, the 2-cycle game has a unique pure Nash equilibrium,
hence the non-deterministic communication complexity of finding 
a Nash or correlated equilibrium of the 2-cycle $N\times N$ game is $O(\log N)$.

The construction of the utility functions of the 2-cycle game was inspired by
the gadget reduction of \cite{RW16} 
which translates inputs of the fixed-point problem in a compact convex space to utility functions. 
However, the utility functions of the 2-cycle game are defined using 
the unique out-neighbor function on a directed graph. 

Our lower bounds are obtained by a direct reduction from the unique set disjointness problem.
We show that the randomized communication complexity of finding the pure Nash equilibrium 
of the 2-cycle $N\times N$ game is $\Omega(N)$
and that given an approximate Nash or correlated equilibrium of the 2-cycle game, 
the players can recover the pure Nash equilibrium with small amount of communication.
Note that for small approximation values, Theorems~\ref{thm:ACE}--\ref{thm:AWSNE} 
are tight for the class of win-lose, sparse games, up to logarithmic factors, 
as a player can send his entire utility function using $O(N\log N)$ bits of communication.

Based on the 2-cycle game, we define a Bayesian game for two players called the Bayesian 2-cycle game.
This is done by splitting the original game into smaller parts, 
where each part corresponds to a type and is in fact a smaller 2-cycle game.
The players are forced to play the same type each time, hence the problem of finding 
an approximate Bayesian Nash equilibrium of this game is essentially reduced to the problem 
of finding an approximate Nash equilibrium of the 2-cycle game.
For a constant $N$, Theorem~\ref{thm:ABNE} gives a \emph{tight} lower bound of $\Omega(T)$ for finding 
a \emph{constant} approximate Bayesian Nash equilibrium of the Bayesian 2-cycle game on $T$ types.

\begin{theorem}\label{thm:ABNE}
Let $N\geq 12$ and $T\geq 2$. Then,
for every $\eps\leq\frac{1}{16N^2}$, every randomized communication protocol 
for finding an $\eps$-approximate Bayesian Nash equilibrium 
of the Bayesian 2-cycle game on $N$ actions and $T$ types, 
with error probability at most $\frac{1}{3}$, 
has communication complexity at least $\Omega(N \cdot T)$.
\end{theorem}

We note that our results do not hold for much larger approximation values, 
since there are examples of approximate equilibria of the 2-cycle game for larger approximation values,
that can be found with small amount of communication (see Appendix~\ref{app:examples} for details).
We discuss some of the remaining open problems in Section~\ref{sec:open}.

\subsection{Related Works}\label{sec:related}
\input{"RelatedWorks.tex"}
\subsection{Proof Overview}\label{sec:overview}
\input{"ProofOverview.tex"}

%% file: RelatedWorks.tex
We overview previous works related to the computation of Nash and correlated equilibria 
of two-player $N\times N$ games. 

\paragraph{Computational complexity.} 
The computational complexity of finding a Nash equilibrium has been extensively studied in literature. 
Papadimitriou \cite{P94} showed that the problem is in \PPAD, and over a decade later it was shown to be complete for that class, even for inverse polynomial approximation values \cite{DGP09,CDT09}. 
However, for constant approximation values, Lipton et al.\ \cite{LMM03} gave a quasi-polynomial time algorithm for finding an approximate Nash equilibrium, and this was shown to be optimal by Rubinstein \cite{R16} under an ETH assumption for \PPAD. 
In stark contrast, exact correlated equilibrium can be computed for two-player games in polynomial time by a linear program \cite{HS89}. 
The decision version of finding Nash and correlated equilibria with particular properties have also been considered in literature (for examples see \cite{GZ89,CS08,ABC11,BL15,DFS16}). 
Finally, we note that Rubinstein \cite{R15} showed that finding a constant approximate Nash equilibrium of Bayesian games is \PPAD-complete.

\paragraph{Query complexity.} 
\cite{FS16} showed a lower bound of $\Omega(N^2)$ on the 
deterministic query complexity of finding an $\eps$-approximate Nash equilibrium, where $\eps<\nicefrac{1}{2}$.
In the other direction, \cite{FGGS15} showed a deterministic query algorithm 
that finds a $\nicefrac{1}{2}$-approximate Nash equilibrium by making $O(N)$ queries. 
For randomized query complexity, \cite{FS16} showed a lower bound of $\Omega(N^2)$ for finding 
an $\eps$-approximate Nash equilibrium, where $\eps<\nicefrac{1}{4N}$.
In the other direction, \cite{FS16} showed a randomized query algorithm 
that finds a $0.382$-approximate Nash equilibrium by making $O(N\log N)$ queries. 
For coarse correlated equilibrium, Goldberg and Roth \cite{GR14} provided a randomized query algorithm that finds 
a constant approximate coarse correlated equilibrium by making $O(N\log N)$ queries.

\paragraph{Communication complexity.} 
The study of the communication complexity of finding a Nash equilibrium
was initiated by Conitzer and Sandholm \cite{CS04}, where they showed that 
the randomized communication complexity of finding a pure Nash equilibrium (if it exists) is $\Omega(N^2)$.
On the other hand, \cite{GP14} showed a communication protocol that finds a 0.438-approximate Nash equilibrium 
by exchanging $\polylog(N)$ bits of communication, 
and \cite{CDFFJS16} showed a communication protocol that finds a 0.382-approximate Nash equilibrium
with similar communication.
%
In a recent breakthrough, Babichenko and Rubinstein \cite{BR16} proved the first lower bound on the communication complexity of finding an approximate Nash equilibrium. They proved that there exists a constant $\eps_0>0$, such that for all $\eps\le \eps_0$, the randomized communication complexity of finding an $\eps$-approximate Nash equilibrium is at least $\Omega(N^{\eps_0})$.
Note that before \cite{BR16} no communication complexity lower bound was known even for finding an exact mixed Nash equilibrium.

%% file: ProofOverview.tex
The 2-cycle $N\times N$ game is defined on two directed graphs, one for each player, 
where both graphs have a common vertex set of size $N$.
The actions of each player are the $N$ vertices.
The utility of a pair of vertices for a player is 1 if he plays the unique out-neighbor 
(according to his graph) of the vertex played by the other player, otherwise it is 0. 

Each graph is constructed from a subset of $\left[\nicefrac{N}{4}\right]$,
such that the two subsets have exactly one element in common. 
The union of the two graphs has a unique 2-cycle that corresponds to the element in the intersection of the subsets.
We show that the 2-cycle game has a unique pure Nash equilibrium, 
that also corresponds to the 2-cycle in the union of the graphs. 
Since it is hard to find the element in the intersection of the subsets, 
finding the pure Nash equilibrium of the 2-cycle game is also hard.

Next, we show that each player can extract from an approximate correlated equilibrium 
a partial mixed strategy on his actions, by looking at the edges of his graph.
We show that the partial strategies are 
either concentrated on the pure Nash equilibrium 
or one of these partial strategies has an unusual probability 
on a vertex which is closely related to the pure Nash equilibrium.

Assuming that the latter does not hold, we show that the partial strategies 
are concentrated on the pure Nash equilibrium as follows:
The union of the two graphs is a layered graph with $\ell=\nicefrac{N}{4}$ layers.
In that graph there is a path of length $\ell-1$ that ends at the 2-cycle.
Using a delicate analysis of the structure of the graph, 
we prove inductively, moving forward along the path, 
that the players play the vertices along the path (up to but not including the 2-cycle) with small probability.
Since the partial strategies hold a meaningful weight of the correlated distribution,
they must be concentrated on the 2-cycle, i.e. the pure Nash equilibrium.

Given these partial strategies, 
the players can recover the pure Nash equilibrium with small amount of communication.
The lower bound then follows from the hardness of finding the pure Nash equilibrium of the 2-cycle game. 

The proof of the lower bound for finding an approximate Nash equilibrium is very similar, 
however it does not require such a delicate analysis.

%% file: OpenProblems.tex
We highlight some open problems.

\begin{enumerate}
\item Coarse correlated equilibrium:
The 2-cycle game has an exact coarse correlated equilibrium $\mu$ defined as follows:
Let $(u_1,v_1)$ and $(u_2,v_2)$ be two arbitrary edges from $G_A$ and $G_B$ respectively, 
such that $v_2\notin N_A(v_1)$ and $v_1\notin N_B(v_2)$. 
Let $\mu(v_1,u_1)=\mu(u_2,v_2)=\nicefrac{1}{2}$.
Note that finding such a distribution requires only small amount of communication.
Therefore, it is not possible to prove non-trivial communication complexity lower bounds for finding a coarse correlated equilibrium of the 2-cycle game.
A natural open problem is to prove any non-trivial bounds on the communication complexity 
of finding a coarse correlated equilibrium of a two-player game.
%

\item Gap amplification:
Our lower bounds hold for inverse polynomial approximation values.
It would be interesting to find a way to amplify the approximation without losing much in the lower bound.
In particular, it is still an open problem to prove a non-trivial communication complexity lower bound
for finding a constant approximate rule correlated equilibrium of a two-player game.
%

\item Multi-player setting:
Babichenko and Rubinstein \cite{BR16} proved an exponential communication complexity lower bound for finding a constant approximate Nash equilibrium of an $n$-player binary action game.
It would be interesting to obtain such an exponential lower bound using techniques that are similar to the ones discussed in this paper, avoiding the simulation theorems and Brouwer function.
Note that exponential lower bounds can not be obtained for finding a correlated equilibrium of multi-player games, 
as there is a protocol for finding an exact correlated equilibrium of $n$-player binary action games 
with $\poly(n)$ bits of communication \cite{HM10, PR08, JL15}.
\end{enumerate}

%% file: Appendix.tex
In this section, we provide trivial approximate equilibria of the 2-cycle game from which it is not possible to recover the disputed index.

\subsection{Approximate Correlated Equilibrium}
Let us suppose that for all $i\in \left[\frac{n}{2}+3\right]$, we have $x_i=y_i=0$.

We define a joint distribution $\mu$ as follows
$$
\mu((i,z_A),(j,z_B))=\begin{cases}
\frac{16\alpha}{n^2}\text{ if }z_A,z_B=0\text{ and }\frac{n}{4}+4\le i,j\le \frac{n}{2}+2,\\
\frac{16\alpha}{n^2}\text{ if }z_A,z_B=0, \frac{n}{4}+2\le j\le \frac{n}{2}+2\text{ and } i=\frac{n}{4}+3,\\
\frac{16\alpha}{n^2}\text{ if }z_A,z_B=0, \frac{n}{4}+2\le i\le \frac{n}{2}+2\text{ and } j=\frac{n}{4}+3,\\
\frac{16\alpha}{n^2}-\frac{64\alpha\cdot (n/4-i+3)}{n^3}\text{ if }z_A,z_B=0, 2\le i,j\le \frac{n}{4}+2\text{ and } i-j=1,\\
\frac{16\alpha}{n^2}-\frac{64\alpha\cdot (n/4-j+3)}{n^3}\text{ if }z_A,z_B=0, 2\le i,j\le \frac{n}{4}+2\text{ and } j-i=1,\\
0\text{ otherwise,}
\end{cases}
$$
where $\alpha$ is some normalizing constant less than 2 such that $\sum_{(u,v)\in V^2}\mu(u,v)=1$.

Let $\varepsilon=64\alpha/n^3$. For every action $u=(i,z_A)$ of Alice such that $z_A\neq 0$, we have that $\mu(u,v)=0$ for all  $v\in V$. Similarly for every action $v=(j,z_B)$ of Bob such that $z_B\neq 0$, we have that $\mu(u,v)=0$ for all  $u\in V$. Also, for every action $u=(i,z_A)$ of Alice such that $i\in \{n/2+3,\ldots ,n\}\cup\{1\}$, we have that $\mu(u,v)=0$ for all  $v\in V$. And, similarly for every action $v=(j,z_B)$ of Bob such that $j\in \{n/2+3,\ldots ,n\}\cup\{1\}$, we have that $\mu(u,v)=0$ for all  $u\in V$. Since $\mu$ is symmetric\footnote{i.e., $\mu(u,v)=\mu(v,u)$ for all $u,v\in V$.}, it follows that in order to show that $\mu$ is an $\varepsilon$-approximate correlated equilibrium we only need to consider a vertex $u=(i,0)$ when $i\in \left[\frac{n}{2}+2\right]$. 

First, we consider when $i\le \frac{n}{4}+2$. Let $u'\in V$. We have
\begin{align*}
\sum_{v\in V}\mu(u,v)\cdot \left( u_A(u',v)- u_A(u,v)\right)&=\mu(u,N_A(u'))-\mu(u,N_A(u))\\
&=\mu(u,N_A(u'))-\frac{16\alpha}{n^2}+\frac{64\alpha\cdot (n/4-i+3)}{n^3}.
\end{align*}
Now if $v=(j,z_B)\in N_A(u')$ and $|j-i|\neq 1$ then, we have $\mu(u,v)=0$. Thus, we assume $j-i=1$, as we suppose $u\neq u'$. Then, we have
\begin{align*}
\mu(u,N_A(u'))&\le\frac{16\alpha}{n^2}-\frac{64\alpha\cdot (n/4-i-1+3)}{n^3}\\
&=\frac{16\alpha}{n^2}-\frac{64\alpha\cdot (n/4-i+3)}{n^3}+\frac{64\alpha}{n^3}.
\end{align*}
This implies,
\begin{align*}
\sum_{v\in V}\mu(u,v)\cdot \left( u_A(u',v)- u_A(u,v)\right)&\le \frac{64\alpha}{n^3} =\varepsilon.
\end{align*}

Next, we consider when $\frac{n}{4}+4\le i\le \frac{n}{2}+2$. Let $u'\in V$. We have
\begin{align*}
\sum_{v\in V}\mu(u,v)\cdot \left(u_A(u',v)- u_A(u,v)\right)&=\mu(u,N_A(u'))-\mu(u,N_A(u))\\
&=\mu(u,N_A(u'))-\frac{16\alpha}{n^2}.
\end{align*}
Now if $v=(j,z_B)\in N_A(u')$ and $j\ge \frac{n}{2}+3$ then, we have $\mu(u,v)=0$. Also if $j\le \frac{n}{4}+2$  then, we have $\mu(u,v)=0$.  Thus, we assume $j \in [n/4+3,n/4+2]$ and $\beta=0$. Then, we have
\begin{align*}
\mu(u,N_A(u'))&\le\frac{16\alpha}{n^2}.
\end{align*}
This implies,
\begin{align*}
\sum_{v\in V}\mu(u,v)\cdot \left(u_A(u',v) -u_A(u,v)\right)&\le 0.
\end{align*}

Finally, we consider when $i= \frac{n}{4}+3$. Let $u'=(i',z_A')\in V$. We have
\begin{align*}
\sum_{v\in V}\mu(u,v)\cdot \left(u_A(u',v)- u_A(u,v)\right)
&=\mu(u,N_A(u'))-\frac{16\alpha}{n^2}+\frac{64\alpha}{n^3}.
\end{align*}

Now if $v=(j,z_B)\in N_A(u')$ and $j\ge \frac{n}{2}+3$ then, we have $\mu(u,v)=0$. Also if $j\le \frac{n}{4}+2$ and $|j-i|\neq 1$ then, we have $\mu(u,v)=0$. Since $u\neq u'$ we have that $j \in [n/4+3,n/4+2]$  and $\beta=0$. Then we have
\begin{align*}
\mu(u,N_A(u'))&\le\frac{16\alpha}{n^2}.
\end{align*}

This implies,
\begin{align*}
\sum_{v\in V}\mu(u,v)\cdot \left(u_A(u',v)- u_A(u,v)\right)&\le \frac{64\alpha}{n^3}=\varepsilon.
\end{align*}

Thus, $\mu$ is  an $\varepsilon$-approximate correlated equilibrium. From Proposition~\ref{ACEtoARCE}, we have that  $\mu$ is  also an $(\varepsilon\cdot N)$-approximate rule correlated equilibrium.

\subsection{Approximate Nash Equilibrium}

Let us suppose that for all $i\in \left[\frac{n}{2}+2\right]$, we have $x_i=y_i=0$.
We define mixed strategies $a,b$ of Alice and Bob respectively as follows
$$a(i,z)=b(i,z)=\begin{cases}
2/n\text{ if }z=0\text{ and }2\le i\le \frac{n}{2}+1\\
0\text{ otherwise}
\end{cases}\ .
$$

Let $\varepsilon=64/N^2$. For every mixed strategy $a'$ for Alice, we have\allowdisplaybreaks
\begin{align*}
\Ex_{u\sim a'}\Ex_{v\sim b} [ u_A(u,v)] - 
\Ex_{u\sim a}\Ex_{v\sim b}[u_A(u,v) ]&=
\left(\sum_{u\in V}b(N_A(u))\cdot a'(u)\right)-\left(\sum_{u\in V}b(N_A(u))\cdot a(u)\right)\\	
&=
\left(\sum_{i\in [n/2]}\frac{2}{n}\cdot a'(i+2,0)\right)-\left(\sum_{i\in [n/2]}\frac{2}{n}\cdot a(i+2,0)\right)\\	
&=
\left(\frac{2}{n}\cdot\sum_{i\in [n/2]} a'(i+2,0)\right)-\left(\frac{4}{n^2}\cdot\frac{n-2}{2}\right)\\	
&\le\left(\frac{2}{n}\cdot 1\right)-\left(\frac{2}{n}-\frac{4}{n^2}\right)\\	
&=\frac{4}{n^2}=\frac{64}{N^2}=\varepsilon
\end{align*}

For every mixed strategy $b'$ for Bob, we have
\begin{align*}
\Ex_{u\sim a}\Ex_{v\sim b'} [ u_B(u,v)] - 
\Ex_{u\sim a}\Ex_{v\sim b}[u_B(u,v) ]&=
\left(\sum_{v\in V}a(N_B(v))\cdot b'(v)\right)-\left(\sum_{v\in V}a(N_B(v))\cdot b(v)\right)\\	
&=
\left(\sum_{i\in [n/2]}\frac{2}{n}\cdot b'(i+2,0)\right)-\left(\sum_{i\in [n/2]}\frac{2}{n}\cdot b(i+2,0)\right)\\	
&=
\left(\frac{2}{n}\cdot\sum_{i\in [n/2]} b'(i+2,0)\right)-\left(\frac{4}{n^2}\cdot\frac{n-2}{2}\right)\\	
&\le\left(\frac{2}{n}\cdot 1\right)-\left(\frac{2}{n}-\frac{4}{n^2}\right)\\	
&=\frac{4}{n^2}=\frac{64}{N^2}=\varepsilon
\end{align*}

Thus, we have that $a$ and $b$ defined above are $64/N^2$-approximate Nash equilibrium.

\subsection{Well Supported Nash Equilibrium}

We define mixed strategies $a,b$ of Alice and Bob respectively as follows
$$a(i,z)=\begin{cases}
1/n\text{ if }z=x_i=0\\
1/n\text{ if }x_i=1\text{ and }z\in \{x_i,x_{i-1}1\}\\
0\text{ otherwise}
\end{cases}\ ,\ \ b(i,z)=\begin{cases}
1/n\text{ if }z=y_i=0\\
1/n\text{ if }y_i=1\text{ and }z\in \{y_i,y_{i-1}1\}\\
0\text{ otherwise}
\end{cases}.
$$

Let $\varepsilon=12/N$. For every action $u\in Supp(a)$ and every action $u'\in V$, 
\begin{align*}
\Ex_{v\sim b} [ u_A(u',v) - u_A(u,v) ]&=\Ex_{v\sim b} [ u_A(u',v)  ]- \Ex_{v\sim b} [  u_A(u,v) ]\\
&= b(N_A(u'))  -  b(N_A(u))\\
&\le  \frac{|N_A(u')|}{n}\\
&\le \frac{12}{N}=\varepsilon
\end{align*}

For every action $v\in Supp(b)$ and every action $v'\in V$, 
\begin{align*}
\Ex_{u\sim a} [ u_B(u,v') - u_B(u,v) ]&=\Ex_{u\sim a} [ u_B(u,v')  ]- \Ex_{u\sim a} [  u_B(u,v) ]\\
&= a(N_B(v'))  -  a(N_B(v))\\
&\le \frac{|N_B(v')|}{n}\\
&\le \frac{8}{N}\le \varepsilon
\end{align*}

Thus, we have that $a$ and $b$ defined above are $8/N$-approximate well supported Nash equilibrium.